\documentclass{lmcs}
\pdfoutput=1

\usepackage{lastpage}
\lmcsdoi{21}{2}{15}
\lmcsheading{}{\pageref{LastPage}}{}{}%
{Dec.~09,~2022}{May~20,~2025}{}

\keywords{model-checking, structural reductions, stutter-sensitivity, omega-automata}

\usepackage[utf8]{inputenc}
\usepackage[T1]{fontenc}
\usepackage{acronym}
\usepackage{dsfont}
\usepackage{color}
\usepackage{epsfig}
\usepackage{graphics}
\usepackage{graphicx}
\usepackage{epstopdf}
\usepackage{url}
\usepackage{pslatex}
\usepackage[english,ruled]{algorithm2e}
\usepackage{newtxtext,newtxmath}

\usepackage{tikz}
\usetikzlibrary{positioning,petri}

\usetikzlibrary{calc,arrows,shapes,automata,petri,positioning,decorations.markings}
\tikzset{
    markedge/.style={
    decoration={ markings,
      mark=at position .5 with {\draw[-,thick,solid,color=red] (-2mm,2mm) -- (2mm,-2mm);\draw[-,thick,solid,color=red] (2mm,2mm) -- (-2mm,-2mm);}  
    },
    postaction={decorate}
  },
}

\usepackage[mode=buildnew]{standalone}
\usepackage{tikz}
\usetikzlibrary{positioning}
\usetikzlibrary{automata, positioning, arrows}
\usepackage{varwidth}
\usetikzlibrary{shapes.geometric}
\usetikzlibrary{arrows.meta}
\usetikzlibrary{arrows}
\usetikzlibrary{arrows,decorations.markings}
\usetikzlibrary{shapes.misc, positioning}

\usepackage{listings}
\lstdefinelanguage{its}
{morekeywords={if, =, +, >=, +, -, (, ), [, ], true, false, typedef, transition, int,GAL, abort, !, \{, \}, label, ", &&, ., composite, synchronization,for,gal, array},
morecomment=[l]{//},
sensitive=false,
}
\newcommand{\nat}%
{\ensuremath{\mathds{N}}}
\newcommand{\zrel}%
{\ensuremath{\mathds{Z}}}
\newcommand{\rea}%
{\ensuremath{\mathds{R}}}
\newcommand{\bool}%
{\ensuremath{\mathds{B}}}

\newcommand{\pre}%
{\ensuremath{\preccurlyeq}}

\usepackage{accents}
\newcommand{\ubar}[1]{\underaccent{\bar}{#1}}

\newcommand{\tuple}[1]{\langle #1 \rangle}

\newcommand{\Places}{\ensuremath{\mathcal{P}}}
\newcommand{\Trans}{\ensuremath{\mathcal{T}}}
\newcommand{\Pre}{\ensuremath{\mathcal{W}_-}}
\newcommand{\Post}{\ensuremath{\mathcal{W}_+}}

\newcommand{\Preset}[1]{\ensuremath{\bullet{#1}}}
\newcommand{\Postset}[1]{\ensuremath{{#1}\bullet}}

\newcommand{\Stutter}{\ensuremath{\mathit{Inv}_\AP}}

\newcounter{rrule}

\renewcommand{\implies}{\Rightarrow}

\renewcommand{\iff}{\Leftrightarrow}

\newcommand{\lang}{\ensuremath{\mathscr{L}}}

\newcommand{\AP}{\ensuremath{\mathrm{AP}}}

\newcommand{\KS}{\ensuremath{\mathit{KS}}}

\newcommand{\tr}{\xrightarrow}

\usepackage{booktabs}
\usepackage{subcaption}

\begin{document}

\title{Structural Reductions and Stutter Sensitive Properties}
\titlecomment{This paper is an extended version of~\cite{forte22}.}



\author[E.~Paviot-Adet]{Emmanuel Paviot-Adet}[a,b]
\author[D.~Poitrenaud]{Denis Poitrenaud\lmcsorcid{0000-0001-9013-4413}}[a,b]
\author[E.~Renault]{Etienne Renault\lmcsorcid{0000-0001-9013-4413}}[c]
\author[Y.~Thierry-Mieg]{Yann Thierry-Mieg\lmcsorcid{0000-0001-7775-1978}}[a]

\address{Sorbonne Université, CNRS, LIP6, Paris, France}	
\email{Emmanuel.Paviot-Adet@lip6.fr, Denis.Poitrenaud@lip6.fr, Yann.Thierry-Mieg@lip6.fr}  

\address{Université Paris Cité, Paris, France}	

\address{EPITA, LRE, Kremlin-Bic\^etre, France}
\email{etienne.renault@epita.fr}

\begin{abstract}

Verification of properties  expressed as $\omega$-regular languages such as LTL can benefit hugely from stutter insensitivity, using a diverse set of reduction strategies. However properties that are not stutter invariant, for instance due to the use of the neXt operator of LTL or to some form of counting in the logic, are not covered by these techniques in general.

We propose in this paper to study a weaker property than stutter insensitivity.
In a stutter insensitive language both adding and removing stutter to
a word does not change its acceptance, any stuttering can be abstracted away; by decomposing this equivalence relation into two implications we obtain weaker conditions.
We define a shortening insensitive language where any word that stutters less than a word in the language must also belong to the language. A lengthening insensitive language has the dual property.
A semi-decision procedure is then introduced to reliably prove shortening insensitive properties or deny lengthening insensitive properties while working with a \emph{reduction} of a system. A reduction has the property that it can only shorten runs. 
Lipton's transaction reductions or Petri net agglomerations are examples of eligible structural reduction strategies.

We also propose to use a partition of a property language into its stutter insensitive, shortening insensitive, lengthening insensitive and length sensitive parts; this lets us apply at least some structural reductions even when working with arbitrary properties.

An implementation and experimental evidence are provided showing that most non-random properties sensitive to stutter are actually shortening or lengthening insensitive. 
Performance of experiments on a large (random) benchmark from the model-checking competition indicates that, despite being a semi-decision procedure, the approach can still improve state of the art verification tools.
\end{abstract}

\maketitle              

\section{Introduction}


Model checking is an automatic verification technique for proving the correctness of systems that have finite state abstractions. 
Properties can for instance be expressed using the popular Linear-time Temporal Logic (LTL)~\cite{Vardi07}.
To verify LTL properties, the automata-theoretic approach~\cite{Vardi07} builds a product between a Büchi automaton representing the negation of the LTL formula and the reachable state graph of the system (seen as a set of infinite runs).
This approach has been used successfully to verify both hardware and software components\cite{KazhamiakinPR04,dureja.18.tacas}, but it suffers from the so called ``state explosion problem'': as the number of state variables in the system increases, the size of the system state space grows exponentially. 

One way to tackle this issue is to consider \textit{structural reductions}. Structural reductions take their roots  in the work of Lipton~\cite{Lipton75} (with extensions by Lamport~\cite{Lamport1989pretending,Lamport90theorem}) and  Berthelot~\cite{Berthelot85}. Nowadays, these reductions are still considered as an attractive way to alleviate the state explosion problem~\cite{Laarman18,berthomieu19,YTM20}. 
Structural reductions strive to fuse structurally ``adjacent'' events into a single atomic step, leading to fewer interleavings of independent events and fewer observable behaviors
in the resulting system. An example of such a structural reduction is shown on Figure~\ref{fig:example} where actions are progressively grouped (see Section~\ref{ss:kripke} for a more detailed presentation). It can be observed that the Kripke structure representing the state space of the program is significantly simplified.

  \begin{figure}[t]
  \centering
\tikzset{
  ->, 
  >=stealth, 
  initial text=$ $,
  every state/.style={fill=gray!10,  inner sep=3pt,minimum size=1pt},
  node distance=3cm
}

\begin{tikzpicture}

   \begin{scope} [xshift=-1.6cm, yshift=0cm]
   \node (th1)  at (2.3,0.3) {\textnormal{Thread $\alpha$}};
   \node (th1)  at (5.,0.3) {\textnormal{Thread $\beta$}};

   \path[-] (3.7,0.3) edge (3.7,-1.5);

   \node[fill=white] (th1)  at (3.7,-0.25) {\textit{Initially, x = y = z = 0}};
   \node (th1)  at (1.75,-0.8) { x++};
   \node (th1)  at (2.4,-1.2) {y=chan.recv()};

   \node (th1)  at (4.3,-0.8) {z=40};
   \node (th1)  at (4.85,-1.2) {chan.send(z)};

   \path[-] (1.2,0) edge (6.2,0);
   \path[-,dashed] (1.2,-0.5) edge (6.2,-0.5);

   \begin{scope}[xshift=0cm,yshift=-.2cm]
     \draw[rounded corners,fill=lightgray!10] (1.2, -1.6) rectangle (6.5, -2.6) {};

     \node[] (th1)  at (3.7,-1.9) { \textit{Atomic propositions: p $\gets$ x $\stackrel{?}{=}$ 0}};
     \node[] (th1)  at (3.7,-2.3) {\small \textit{\hspace{22.2ex} q $\gets$ y $\stackrel{?}{=}$ 0}};

   \end{scope}
   
    \node (th1)  at (3.7,-3.6) {(1) Program};

     \node (0)  at (.9,-.68) {\textcolor{lightgray}{\textit{\tiny{ $\alpha_0\rightarrow$}}}};
     \node (0)  at (.9,-1.03) {\textcolor{lightgray}{\textit{\tiny{  $\alpha_1\rightarrow$}}}};
     \node (0)  at (.9,-1.35) {\textcolor{lightgray}{\textit{\tiny{  $\alpha_2\rightarrow$}}}};

     \node (0)  at (6.7,-.68) {\textcolor{lightgray}{\textit{\tiny{$\leftarrow\beta_0$}}}};
     \node (0)  at (6.7,-1.03) {\textcolor{lightgray}{\textit{\tiny{$\leftarrow\beta_1$}}}};
     \node (0)  at (6.7,-1.35) {\textcolor{lightgray}{\textit{\tiny{$\leftarrow\beta_2$}}}};

   \end{scope}

   \begin{scope}[xshift=7cm,yshift=0cm]
   \node[state, initial, initial left] (s0) at  (0,-1.5)   {$p q$};
   \node[state]                         (s1) at  (0, 0) {$\bar pq$};
   \node[state]                         (s2) at  (2.2, 0) {$\bar pq$};
   \node[state]                         (s3) at  (4.4, 0) {$\bar pq$};
   \node[state]                         (s4) at  (2.2, -1.5) {$pq$};
   \node[state]                         (s5) at  (4.4, -1.5) {$pq$};
   \node[state]                         (s6) at  (6.6, -1.5) {$\bar p \bar q$};
   \draw   (s0) edge[] node[fill=white]{{x++}} (s1);
   \draw   (s1) edge[bend left] node[above]{{z=40}} (s2);
   \draw   (s2) edge[bend left] node[above]{{chan.send(z)}} (s3);
   \draw   (s0) edge[bend right] node[below]{{z=40}} (s4);
   \draw   (s4) edge[] node[fill=white]{{x++}} (s2);
   \draw   (s4) edge[bend right] node[below]{{chan.send(z)}} (s5);
   \draw   (s3) edge[bend left, below] node[fill=white]{\hspace{3ex}{y=chan.recv()}} (s6);
    \draw   (s5) edge[] node[fill=white]{{x++}} (s3);

    \node (th1)  at (2.6,-3.) {$\lang_A = \{$ $pq^3$ $\bar p q$ $\bar p \bar q ^\omega$, 
      $pq^2$ $\bar p q^2$ $\bar p \bar q ^\omega$,
      $pq$ $\bar p q^3$ $\bar p \bar q ^\omega$
    $\} $};

    \node (th1)  at (2.4,-3.6) {(2) State space if each statement of the program is atomic};

   \end{scope}

   \begin{scope}[xshift=0cm,yshift=-5cm]
   \node[state, initial, initial left] (s0) at  (0,-1.5)   {$p q$};
   \node[state]                        (s1) at  (0, 0) {$\bar pq$};
   \node[state]                        (s3) at  (3.4, 0) {$\bar pq$};
   \node[state]                         (s5) at  (3.4, -1.5) {$pq$};
   \node[state]                         (s6) at  (5.6, -1.5) {$\bar p \bar q$};
   \draw   (s0) edge[] node[fill=white]{{x++}} (s1);
   \draw   (s1) edge[red, dashed, bend left,below] node[]{{\hspace{2ex} \shortstack[l]{z=40;\\chan.send(z)}}} (s3);
   \draw   (s0) edge[red, dashed, bend right,above] node[]{{\hspace{2ex} \shortstack[l]{z=40;\\chan.send(z)}}} (s5);
   \draw   (s3) edge[bend left, below] node[fill=white]{\hspace{3ex}{y=chan.recv()}} (s6);
      \draw   (s5) edge[] node[fill=white]{{\hspace{-2ex} x++}} (s3);

       \node (th1)  at (2.1,-3.) {$\lang_B =\{$ $pq$ $\bar p q^2$ $\bar p \bar q ^\omega$, $pq^2$ $\bar p q$ $\bar p \bar q ^\omega$ $\}$ };

    \node (th1)  at (2.1,-3.6) {(3) State space if "z=40;send(z)" is atomic.};

   \end{scope}

   \begin{scope}[xshift=8.5cm,yshift=-5cm]
   \node[state, initial, initial left] (s0) at  (0,-1.5)   {$p q$};
   \node[state]                        (s1) at  (0, 0) {$\bar pq$};
   \node[state]                         (s6) at  (4.6, -1.5) {$\bar p \bar q$};
   \draw   (s0) edge[] node[fill=white]{{x++}} (s1);
   \draw   (s1) edge[red, thick, dashed, bend left, below] node[red, fill=none,align=left]{{\hspace{2ex} \shortstack[l]{z=40;\\chan.send(z);\\y=chan.recv()}}} (s6);
   
       \node (th1)  at (2.1,-3.) {$\lang_C = \{$  $pq$ $\bar p q$ $\bar p \bar q ^\omega$ $\}$ };

     \node (th1)  at (2.4,-3.6) {(4) State space if "z=40;send(z);y=chan.recv()" is atomic.};
 
   \end{scope}

\end{tikzpicture}
    \caption{
    Example of reductions. (1) is a program with two threads and 3 variables. \textit{chan} is a communication channel where \textit{send(int)} insert a message and \textit{int recv()} waits until a message is available and then consumes it. We assume that one can only observe whether $x$ or $y$ is zero denoted $p$ and $q$. (2) depicts the state-space represented as a Kripke structure. Each node is labelled by the values of atomic propositions $p$ and $q$. 
    When an instruction is executed the values of these propositions \emph{may} evolve. (3) represents the state-space of a structurally reduced version of the program where actions of thread $\beta$ "z=40;chan.send(z)" are fused into a single atomic operation. (4) represents the state-space of a program where the three actions of the original program "z=40;chan.send(z);y=chan.recv()" are now a single atomic step.
    }
    \label{fig:example}
\end{figure}

%
%

Traditionally structural reductions construct a smaller system that
preserves properties such as deadlock freedom, liveness, reachability~\cite{HPP06}, and \emph{stutter insensitive} temporal logic~\cite{PPP00} such as the fragment of LTL without the next operator LTL$_{\setminus X}$. 
The verification of a \emph{stutter insensitive} property on a given system
does not depend on whether non observable events (i.e. that do not update atomic propositions)
are abstracted or not.
On Fig~\ref{fig:example} both instructions "$z=40;$" and "$chan.send(z)$" of thread $\beta$ are non observable.

This paper shows that structural reductions can in fact be used even for fragments of LTL that are \emph{not} stutter insensitive. 
We identify two fragments that we call 
\emph{shortening insensitive} (if a word is in the language, then so is any version that stutters less) or \emph{lengthening insensitive} (if a word is in the language, then so is any version that stutters more).
Based on this classification we introduce two semi-decision procedures that provide a reliable verdict only in one direction: e.g. presence of counter examples is reliable for lengthening insensitive properties, but absence is not.
We also propose to decompose properties to partition their language into stutter insensitive, shortening insensitive and lengthening insensitive parts, allowing one to use structural reductions even if the property is not globally insensitive to length.

The paper is structured as follows, Section~\ref{sec:defs} presents the definitions and notations relevant to our setting in an abstract manner, focusing on the level of description of a language. Section~\ref{sec:red} instantiates these definitions in the more concrete setting of LTL verification. Section~\ref{sec:perfs} provides experimental evidence supporting the claim that the method we propose is both 
applicable to many formulae and can significantly improve state of the art model-checkers. 
Section~\ref{sec:extension} introduces several extensions to the approach, providing new material with respect to~\cite{forte22}.
Some related work is presented in Section~\ref{sec:related} before concluding.

\begin{figure}[t]
    \centering
\includestandalone[width=0.7\linewidth]{fig/figure}
    \caption{
    $\Sigma^\omega$ is represented as a circle that is partitioned into equivalence classes of words ($\hat{r_0}, \hat{r_1} \ldots$). Each point in the space is a word, and some of the $\pre$ relations are represented as arrows ; the red point is the shortest word $\underline{\hat{r}}$ in the equivalence class. Gray areas are inside the language, white are outside of it. Four languages are depicted: 
   $\lang_a$: equivalence classes are entirely inside or outside \textbf{a stutter insensitive language}, 
    $\lang_b$: the "bottom" of an equivalence class may belong to \textbf{a shortening insensitive language},  
    $\lang_c$: the "top" of an equivalence class may belong to \textbf{a lengthening insensitive language},  
    $\lang_d$: some languages are neither lengthening insensitive nor shortening insensitive.
    }
    \label{fig:lang}
\end{figure}

\section{Definitions}
\label{sec:defs}

In this section we first introduce in Section~\ref{ss:shorterthan} a \emph{``shorter than''} partial order relation on infinite words, based on the number of repetitions or stutter in the word. 
This partial order gives us in Section~\ref{ss:shortlong} the notions of shortening and lengthening insensitive language, which are shown to be weaker versions of classical stutter insensitivity in Section~\ref{ss:relationSI}.
We then define in Section~\ref{ss:checkred} the \emph{reduction of a language} which contains a shorter representative of each word in the original language. 
Finally we show that we can use a semi-decision procedure to verify shortening or lengthening insensitive properties using a reduction of a system.

\subsection{A ``Shorter than'' relation for infinite words}
\label{ss:shorterthan}

\begin{defi}
    [Word] A word over a finite alphabet $\Sigma$ is
    an infinite sequence of symbols in $\Sigma$.
    We canonically denote a word $r$ using one of the two forms:  
    \begin{itemize}
        \item (plain word) $r=a_0^{n_0}a_1^{n_1}a_2^{n_2}\ldots$ with for all $i \in \nat$, $a_i \in \Sigma$, $n_i \in \nat^+$ and $a_i \neq a_{i+1}$, or
        \item ($\omega$-word) $r=a_0^{n_0}a_1^{n_1}\ldots a_k^\omega$ with $k \in \nat$ and for all $0 \leq i \leq k$, $a_i \in \Sigma$, and for $i<k$, $n_i \in \nat^+$ and $a_i \neq a_{i+1}$. $a_k^\omega$ represents an infinite stutter on the final symbol $a_k$ of the word.
    \end{itemize}
    The set of all words over alphabet $\Sigma$ is denoted $\Sigma^\omega$.
\end{defi}

These notations using a power notation for repetitions of a symbol in a word are introduced to highlight stuttering. We force the symbols to alternate to ensure we have a canonical representation: with $\sigma$ a suffix (not starting by symbol $b$), the word $aab\sigma$ must be represented as $a^2b^1\sigma$ and not $a^1a^1b^1\sigma$. 
To represent a word of the form $aabbcccccc\ldots$ we use an $\omega$-word: $a^2b^2c^\omega$.

\begin{defi}
    [Shorter than] A plain word  $r=a_0^{n_0}a_1^{n_1}a_2^{n_2}\ldots$ is \emph{shorter} than a plain word $r'=a_0^{n_0'}a_1^{n_1'}a_2^{n_2'}\ldots$ if and only if for all $i \in \nat$, $0 < n_i \leq n_i'$. 
    For two $\omega$-words $r=a_0^{n_0}\ldots a_k^\omega$ and $r'=a_0^{n_0'}\ldots a_k^\omega$, $r$ is \emph{shorter} than $r'$ if and only if for all $i < k$, $n_i \leq n_i'$.
    
    We denote this relation on words as $r \pre r'$. 
\end{defi}

For instance, for any given suffix $\sigma$, $a b \sigma \pre a^2 b \sigma$. Note that $a b \sigma \pre a b^2 \sigma$ as well, but that $a^2 b \sigma$ and $a b^2 \sigma$ are incomparable. $\omega$-words are incomparable with plain words. 

\begin{pty}
    The $\pre$ relation is a partial order on words. 
\end{pty}

\begin{proof}
The relation is clearly reflexive ($\forall r \in \Sigma^\omega, r \pre r$), anti-symmetric ($\forall r, r' \in \Sigma^\omega$, $r \pre r' \land r' \pre r \implies r = r'$) and transitive ($\forall r, r',r'' \in \Sigma^\omega, r \pre r' \land r' \pre r'' \implies r \pre r''$). The order is partial since some words (such as $a^2 b \sigma$ and $a b^2 \sigma$ presented above) are incomparable.
\end{proof}

\begin{defi}[Stutter equivalence]
\label{def:stutterequiv}
    A word  $r$ is \emph{stutter equivalent} to $r'$, denoted as $r \sim r'$ if and only if there exists a shorter word $r''$ such that $r'' \pre r \land r'' \pre r'$. This relation $\sim$ is an equivalence relation thus partitioning words of $\Sigma^\omega$ into equivalence classes.
    
    We denote $\hat{r}$ the equivalence class of a word $r$ and denote $\ubar{\hat{r}}$ the shortest word in that equivalence class.
\end{defi}

For any given word $r=a_0^{n_0}a_1^{n_1}a_2^{n_2}\ldots$ there is a shortest representative in $\hat{r}$ that is the word $\ubar{\hat{r}}=a_0 a_1 a_2\ldots$ where no symbol is ever consecutively repeated more than once (until the $\omega$ for an $\omega$-word). By definition all words that are comparable to $\ubar{\hat{r}}$ are stutter equivalent to each other, since $\ubar{\hat{r}}$ can play the role of $r''$ in the definition of stutter equivalence, giving us an equivalence relation: it is reflexive, symmetric and transitive.

For instance, with $\sigma$ denoting a suffix, $\ubar{\hat{r}}=a b \sigma$ would be the shortest representative of any word of $\hat{r}$ of the form $a^{n_0}b^{n_1}\sigma$. We can see by this definition that, despite being incomparable, $a^2 b \sigma \sim a b^2 \sigma$ since $a b \sigma \pre a^2 b \sigma$ and $a b \sigma \pre a b^2 \sigma$.

\subsection{Sensitivity of a language to the length of words}
\label{ss:shortlong}

\begin{defi}
    [Language] A language $\lang$ over a finite alphabet $\Sigma$ is a set of words over $\Sigma$, hence $\lang \subseteq \Sigma^\omega$. We denote by $\bar{\lang} =  \Sigma^\omega \setminus \lang$ the complement of a language $\lang$.
\end{defi}

In the literature, most studies that exploit a form of stuttering are focused on \emph{stutter insensitive} languages \cite{porBook,Valmari90,Peled94,GodefroidW94,HPP06}.
 In a stutter insensitive language $\lang$, duplicating any letter (also called stuttering) or removing any duplicate letter from a word of $\lang$ must produce another word of $\lang$.
In other words, all stutter equivalent words in a class $\hat{r}$ must be either in the language or outside it. Let us introduce weaker variants of this property, which were originally presented in~\cite{forte22}.

\begin{defi}[Shortening insensitive]
\label{def:shortins}
    A language $\lang$ is \emph{shortening insensitive} if and only if for any word $r$ it contains, all shorter words $r' \in \Sigma^\omega$ such that $r' \pre r$ also belong to $\lang$.
    $$\forall r \in \lang, \forall r' \in \Sigma^\omega, r' \pre r \implies r' \in \lang$$ 
\end{defi}

For instance, a shortening insensitive language $\lang$ that contains the word $a^3 b \sigma$ must also contain shorter words $a^2 b \sigma$, and $a b \sigma$. If it contains $a^2 b^2 \sigma$ it also contains $a^2 b \sigma$, $a b^2 \sigma$ and $a b \sigma$.

\begin{defi}[Lengthening insensitive]
\label{def:longins} A language $\lang$ is \emph{lengthening insensitive} if and only if for any word $r$ it contains, all longer words $r' \in \Sigma^\omega$ such that $r \pre r'$ also belong to $\lang$. 
    $$\forall r \in \lang, \forall r' \in \Sigma^\omega, r \pre r' \implies r' \in \lang$$
\end{defi}

For instance, a lengthening insensitive language $\lang$ that contains the word $a^2 b \sigma$ must also contain all longer words $a^3 b \sigma$, $a^2 b^2 \sigma$ \ldots, and more generally words of the form $a^{n} b^{n'} \sigma$ with $n \geq 2$ and $n' \geq 1$. If it contains $\ubar{\hat{r}}=a b \sigma$ the shortest representative of an equivalence class, it contains all words in the stutter equivalence class.

While stutter insensitive languages have been heavily studied, there is, to our knowledge, no study on what reductions are possible if only one direction holds, i.e. the language is shortening or lengthening insensitive, but not both. A shortening insensitive language is essentially asking for something to happen \emph{before} a certain deadline or stuttering ``too much''. A lengthening insensitive language is asking for something to happen \emph{at the earliest} at a certain date or after having stuttered at least a certain number of times. Figure~\ref{fig:lang} represents these situations graphically.


\subsection{Relationship to stutter insensitive logic}
\label{ss:relationSI}

A language is both shortening and lengthening insensitive if and only if it is stutter insensitive (see Fig.~\ref{fig:lang}).
This fact is already used in~\cite{ADL15} to identify a stutter insensitive language using only its automaton.
Furthermore since stutter equivalent classes of words are entirely inside or outside a stutter insensitive language, 
a language $\lang$ is stutter insensitive if and only if the complement language $\bar{\lang}$ is stutter insensitive.

However, if we look at sensitivity to length and how it interacts with the complement operation, we find a dual relationship where the complement of a shortening insensitive language is lengthening insensitive and vice versa. 

\begin{pty}
\label{prop:shortduallong}
A language $\lang$ is shortening insensitive if and only if the complement language $\bar{\lang}$ is lengthening insensitive.
\end{pty}

\begin{proof}
Let $\lang$ be shortening insensitive. Let $r \in \bar{\lang}$ be a word in the complement of $\lang$. Any word $r'$ such that $r \pre r'$ must also belong to $\bar{\lang}$, since if it belonged to the shortening insensitive $\lang$, $r$ would also belong to $\lang$. Hence $\bar{\lang}$ is lengthening insensitive. The converse implication can be proved using the same reasoning.
\end{proof}

If we look at Figure~\ref{fig:lang}, the dual effect of complement on the sensitivity of the language to length is  apparent: if gray and white are switched we can see that $\bar{\lang_b}$ is lengthening insensitive and $\bar{\lang_c}$ shortening insensitive.

\subsection{When is visiting shorter words enough?}
\label{ss:checkred}

\begin{defi}[Reduction]
\label{def:reduction}
    Let $I$ be a reduction function $\Sigma^\omega \rightarrow \Sigma^\omega$ such that $I(r) \pre r$ for every $r$ in $\Sigma^\omega$. The reduction by $I$ of a language $\lang$ is $\mathit{Red}_I(\lang)=\{I(r) \mid r \in \lang\}$. 
\end{defi}


Note that the $\pre$ partial order is not strict so that the image of a word may be the word itself, hence the identity function is a reduction function. In most cases however we expect the reduction function to map many words $r$ of the original language to a single shorter word $r'$ of the reduced language. 
Note that given any two reduction functions $I$ and $I'$, $I \circ I'$ is also a reduction function, therefore $\mathit{Red}_I(\mathit{Red}_{I'}(\lang)) = \mathit{Red}_{I\circ I'}(\lang)$ is still a reduction of $\lang$. Hence chaining reduction rules still produces a reduction.
As we will discuss in Section~\ref{ss:kripke} structural reductions of a specification such as Lipton's transaction reduction~\cite{Lipton75,Laarman18} or Petri net agglomerations~\cite{Berthelot85,YTM20} (see also Section~\ref{sec:petri}) induce a reduction at the language level.
In Fig.~\ref{fig:example} fusing statements into a single atomic step in the program induces a reduction of the language.

\begin{thm}
[Reduced Emptiness Checks]
\label{th:short}
Let $I$ be a reduction function. Given two languages $\lang$ and $\lang'$, 
\begin{itemize}
    \item if $\lang$ is shortening insensitive, then $\lang \cap \mathit{Red}_I(\lang') = \emptyset \implies \lang \cap \lang' = \emptyset$
    \item if $\lang$ is lengthening insensitive, then $\lang \cap \mathit{Red}_I(\lang') \neq \emptyset \implies \lang \cap \lang' \neq \emptyset$.
\end{itemize}
\end{thm}


\begin{proof}
(Shortening insensitive $\lang$) Assume that $\lang \cap \mathit{Red}_I(\lang') = \emptyset$ so there does not exists $r' \in \lang \cap \mathit{Red}_I(\lang')$. Because $\lang$ is shortening insensitive, it is impossible that any word $r$ with $r' \prec r$ belongs to $\lang \cap \lang'$.
(Lengthening insensitive $\lang$) At least one word $r'$ is in $\lang$ and $\mathit{Red}_I(\lang')$. Let $r'=I(r)$ with $r' \in \lang$, since $\lang$ is lengthening insensitive and $r' \pre r$, $r$ must also belong to $\lang$ and thus to $\lang \cap \lang'$.
\end{proof}

With this theorem original to this paper we now can build a semi-decision procedure that is able to prove \emph{some} lengthening or disprove \emph{some} shortening insensitive properties using a reduction of a system. In practice, the language $\lang$ will be the language of the negation of a property, and $\lang'$ and  $\mathit{Red}_I(\lang')$ will be respectively the language of a system and the language of its reduced version.

\section{Application to Verification}
\label{sec:red}
We now introduce the more concrete setting of LTL verification to exploit the theoretical results on languages and their shortening/lengthening sensitivity developed in Section~\ref{sec:defs}.

\subsection{Kripke Structure}
\label{ss:kripke}

From the point of view of LTL verification with a state-based logic, executions of a system (also called \emph{runs}) are seen as infinite words over the alphabet $\Sigma=2^{\AP}$, where $\AP$ is a set of atomic propositions that may be true or false in each state. So each symbol in a run gives the truth value of all of the atomic propositions in that state of the execution, and each time an action happens we progress in the run to the next symbol. Some actions of the system update the truth value of atomic propositions, but some actions can leave them unchanged, which corresponds to stuttering. 

\begin{defi}
[Kripke Structure Syntax] Let $\AP$ designate a set of atomic propositions. 
A Kripke structure $\KS_\AP=\tuple{S,R,\lambda,s_0}$ over $\AP$ is a tuple where $S$ is the set of states, $R \subseteq S \times S$ is the transition relation, $\lambda: S \rightarrow 2^\AP$ is the state labeling function, and $s_0 \in S$ is the initial state.
\end{defi}
\begin{defi}[Kripke Structure Semantics]
\label{def:kripke}
The language $\lang(\KS_\AP)$ of a Kripke structure $\KS_\AP$ is defined over the alphabet $2^\AP$.
It contains all runs of the form $r=\lambda(s_0) \lambda(s_1) \lambda(s_2)\ldots$ where $s_0$ is the initial state of $\KS_\AP$ and $\forall i \in \nat$, either $(s_i,s_{i+1}) \in R$, or if $s_i$ is a deadlock state such that $\forall s' \in S, (s_i,s') \not\in R$ then $s_{i+1}=s_i$.
\end{defi}

All system executions are considered maximal, so that they are represented by infinite runs. If the system can deadlock or terminate in some way, we can extend these finite words by an infinite stutter on the last symbol of the word to obtain a run.

\medskip
\textbf{Example.}
Subfigure (1) of Figure~\ref{fig:example} depicts a program where each thread ($\alpha$ and $\beta$) has  three reachable positions (we consider that each instruction is atomic).
In this example we assume that the only observable atomic propositions are $p$ (true when $x=0$) and $q$ (true when $y=0$). The variable $z$ is not observed.

Subfigure (2) of Figure~\ref{fig:example} depicts the reachable states of this system as a Kripke structure.
Actions of thread $\beta$ (which do not modify the value of $p$ or $q$) are horizontal while actions of thread $\alpha$ are vertical. While each thread has 3 reachable positions, the emission of the message by $\beta$ must precede the reception by $\alpha$ so that some situations are unreachable.
Based on Definition~\ref{def:kripke} we can compute the language $\lang_A$ of this system. It consists of three words: when thread $\beta$ goes first $pq^3$ $\bar p q$ $\bar p \bar q ^\omega$, with an interleaving $pq^2$ $\bar p q^2$ $\bar p \bar q ^\omega$, and when thread $\alpha$ goes first $pq$ $\bar p q^3$ $\bar p \bar q ^\omega$.

In subfigure (3) of Figure~\ref{fig:example}, 
the actions "$z=40;chan.send(z);$"of thread $\beta$ are fused into a single atomic operation.
This is possible because action $z=40$ of thread $\beta$ is stuttering (it cannot affect either $p$ or $q$) and is non-interfering with other events (it neither enables nor disables any event other than subsequent instruction "chan.send(z)"). 
 The language of this smaller KS is a reduction of the language of the original system. It contains two runs: thread $\alpha$ goes first  $pq$ $\bar p q^2$ $\bar p \bar q ^\omega$ and thread $\beta$ goes first $pq^2$ $\bar p q$ $\bar p \bar q ^\omega$.

In subfigure (4) of Figure~\ref{fig:example}, 
the already fused action "$z=40;chan.send(z);$"of thread $\beta$ is further fused with the \textit{chan.recv();} action of thread $\alpha$.
This leads to a smaller KS whose language is still a reduction of the original system now containing a single run: $pq$ $\bar p q$ $\bar p \bar q ^\omega$.
This simple example shows the power of structural reductions, when they are applicable, with a reduction of the initial language to a single word.

\subsection{Automata theoretic verification}

Let us consider the problem of model-checking of an $\omega$-regular  property $\varphi$ (such as one described by an LTL formula) on a system using the automata-theoretic approach~\cite{Vardi07}. In this approach, we wish to answer the problem of language inclusion: do all runs of the system $\lang(\KS)$ belong to the language of the property $\lang(\varphi)$ ? To do this, when the property $\varphi$ is an $\omega$-regular language (e.g. an LTL or PSL formula), we first negate the property $\lnot\varphi$, then build a (variant of) a Büchi automaton $A_{\lnot\varphi}$ whose language\footnote{Because computing the complement $\bar{A}$ of a Büchi automaton $A$ is worse than exponential in the worst case~\cite{yan.08.lmcs,schewe.12.atva}, syntactically negating $\varphi$ and producing an automaton $A_{\lnot \varphi}$ is preferable when $A$ is derived from e.g. an LTL formula.} 
 consists of all the runs that do not satisfy $\varphi$ i.e. $\lang(A_{\lnot\varphi}) = \Sigma^\omega \setminus \lang(\varphi)$.
We then perform a synchronized product between this Büchi automaton and the Kripke structure $\KS$ corresponding to the system's state space $A_{\lnot\varphi} \otimes \KS$ (where $\otimes$ is defined to satisfy $\lang(A \otimes B)=\lang(A) \cap \lang(B)$). Either the language of the product is empty $\lang(A_{\lnot\varphi} \otimes \KS) = \emptyset$, and the property $\varphi$ is thus true of this system, or the product is non empty, and from any run in the language of the product we can build a counter-example to the property.

We will consider in the rest of the paper that the shortening or lengthening  insensitive language of Definitions~\ref{def:shortins} and~\ref{def:longins} is given as an omega-regular language or Büchi automaton typically obtained from the negation of an LTL property, and that the reduction of Definition~\ref{def:reduction} is applied to a language that corresponds to all runs in a Kripke structure typically capturing the state space of a system.

\medskip
\textbf{LTL verification with reductions.} With Theorem~\ref{th:short}, a shortening insensitive property shown to be true on the reduction (empty intersection with the language of the negation of the property) is also true of the original system. A lengthening insensitive property shown to be false on the reduction (non-empty intersection with the language of the negation of the property, hence counter-examples exist) is also false in the original system.
Unfortunately, our procedure cannot prove using a reduction that a shortening insensitive property is false, or that a lengthening insensitive property is true. We offer a semi-decision procedure.

\subsection{Detection of language sensitivity}
\label{sec:detection}

We now present a strategy to decide if a given property expressed as a Büchi automaton is shortening insensitive, lengthening insensitive, or both.

This section relies heavily on the operations introduced and discussed at length in~\cite{ADL15}. 
The authors define two \emph{syntactic} transformations $sl$ and $cl$ of a transition-based generalized Büchi automaton (TGBA) $A_\varphi$ that can be built from any LTL formula $\varphi$ to represent its language $\lang(\varphi)=\lang(A_\varphi)$~\cite{CouvreurFM99}. TGBA are a variant of Büchi automata where the acceptance conditions are placed on edges rather than states of the automaton. 

The $cl$ closure operation \emph{decreases stutter}, it adds to the language any word $r' \in \Sigma^\omega$ that is shorter than a word $r$ in the language. Informally, the strategy consists in detecting when a sequence $q_1 \tr{a} q_2 \tr{a} q_3$ is possible and adding an edge $q_1 \tr{a} q_3$, hence its name $cl$ for ``closure''. 
The $sl$ self-loopization operation \emph{increases stutter}, it adds to the language any run $r' \in \Sigma^\omega$ that is longer than a run $r$ in the language.
Informally, the strategy consists in adding a self-loop to any state so that we can always decide to repeat a letter (and stay in that state) rather than progress in the automaton, hence its name $sl$ for ``self-loop''. 
More formally $\lang(cl(A_\varphi))=\{r' \mid \exists r \in \lang(A_\varphi), r' \pre r \}$ and $\lang(sl(A_\varphi))=\{r' \mid \exists r \in \lang(A_\varphi), r \pre r' \}$.

Using these operations \cite{ADL15} shows that there are several possible ways to test if an omega-regular language (encoded as a Büchi automaton) is stutter insensitive: 
essentially applying either of the operations $cl$ or $sl$ should leave the language unchanged. This allows one to recognize that a property is stutter insensitive even
though it syntactically contains e.g. the neXt operator of LTL.

For instance $A_\varphi$ is stutter insensitive if and only if $\lang(sl(cl(A_\varphi)) \otimes A_{\lnot \varphi}) =\emptyset$.
The full test is thus simply reduced to a language emptiness check testing that both $sl$ and $cl$ operations
leave the language of the automaton unchanged.

Indeed for stutter insensitive languages, all or none of the runs belonging to a given stutter equivalence class of runs $\hat{r}$ must belong to the language $\lang(A_\varphi)$. 
In other words, if shortening or lengthening a run can make it switch from belonging to $A_\varphi$ to belonging to $A_{\lnot \varphi}$, the language is stutter sensitive. 
This is apparent on Figure~\ref{fig:lang}.

We want weaker conditions here, but we can reuse the $sl$ and $cl$ operations developed for testing stutter insensitivity.
Indeed for an automaton $A$ encoding a shortening insensitive language, $\lang(cl(A))=\lang(A)$ should hold.
Conversely if $A$ encodes a lengthening insensitive language, $\lang(sl(A))=\lang(A)$ should hold.
We express these tests as emptiness checks on a product in the following way.

\begin{thm}
\label{th:sitest}
[Testing sensitivity]
Let $A$ designate a Büchi automaton, and $\bar{A}$ designate its complement.

$\lang(cl(A) \otimes \bar{A}) = \emptyset$, if and only if $A$ defines a shortening insensitive language.

$\lang(sl(A) \otimes \bar{A}) = \emptyset$ if and only if $A$ defines a lengthening insensitive language.
\end{thm}

\begin{proof}
The expression $\lang(cl(A) \otimes \bar{A}) = \emptyset$ is equivalent to $\lang(cl(A))=\lang(A)$. The lengthening insensitive case is similar.
\end{proof}

Thanks to property~\ref{prop:shortduallong}, and in the spirit of~\cite{ADL15} we could also test the complement of a language for the dual property if that is more efficient, i.e. $\lang(sl(\bar{A}) \otimes A) = \emptyset$ if and only if $A$ defines a shortening insensitive language and similarly $\lang(cl(\bar{A}) \otimes A) = \emptyset$ iff $A$ is lengthening insensitive. We did not really investigate these alternatives as the complexity of the test was already negligible in all of our experiments. 

Overall complexity of these tests is dominated by the complement operation on the automaton $A$ to compute $\bar{A}$; this operation is worst case exponential in the size of $A$. However, when the automaton $A_\varphi$ is obtained by translating a LTL formula $\varphi$, we can compute $A_{\lnot\varphi}$ the automaton for the negation of $\varphi$ instead of the complement $\bar{A_\varphi}$, avoiding this exponential. The $cl$ operation only introduces new edges, so $cl(A)$ is worst case quadratic in the number of states of $A$. The $sl$ operation unfortunately may require adding new states, to allow to stutter on the last seen valuation of the system (a valuation is in $2^{AP}$), hence is worst case exponential over the number of atomic propositions in $A$, but formulas in our experiments do not have more than 5 atomic propositions so this complexity does not dominate the procedure. The emptiness check computes the product of two automata and which is worst case linear to the product of the sizes of the two automata.


\subsection{Semi Decision Procedure}

With these elements we can describe a semi decision procedure to verify a property $\varphi$ on a system represented by a Kripke structure $\KS$.

\begin{enumerate}
\item Decide if $\varphi$ is shortening insensitive or lengthening insensitive using Theorem~\ref{th:sitest}. If neither is true, abort the procedure.
\item Reduce the system $\KS$ using the atomic propositions of $\varphi$ as observed alphabet $\KS'=\mathit{Reduce}(\KS,\varphi)$. Thus $\KS'$ is a reduction in the sense of Definition~\ref{def:reduction}.
\item Use a model-checker to verify whether the reduced system $\KS'$ satisfies $\varphi$, i.e. whether $\lang(A_{\lnot\varphi} \otimes \KS')=\emptyset$.
\item If $\KS'$ satisfies $\varphi$ and $\varphi$ is \emph{shortening insensitive} or $\KS'$ does not satisfy $\varphi$ and $\varphi$ is \emph{lengthening insensitive}, conclude. 
\item Otherwise abort the procedure, and run a model-checker using the original system and property.
\end{enumerate}

Note that the complexity of this procedure is dominated by the model-checking step, since $\KS$ is typically extremely large in front of $A_{\lnot\varphi}$. If the reduction is successful, $\KS'$ can be up to exponentially smaller than $\KS$, but this depends on the size of the observed alphabet as well as on the system definition. The experiments presented in Section~\ref{sec:perfs} confirm that in most cases verifying on the reduced system is much cheaper than doing model-checking on the original system.

\subsection{Agglomeration of events produces shorter runs}
\label{sec:aggloShort}
Structural reductions are one of the possible strategies to reduce the complexity of analyzing a system.
Depending on the input formalism the terminology used is different, but the main results remain stable.

In~\cite{Lipton75} \emph{transaction reduction} consists in fusing two adjacent actions of a thread (or even across threads in recent versions such as~\cite{Laarman18} ). 
The first action must not modify atomic properties and must commutate with any action of other threads. Fusing these actions leads to shorter runs, where a stutter is lost. 
In the program of Fig.~\ref{fig:example}, "z=40" is enabled from the initial state and must happen before "chan.send(z)", but it commutes with instructions of thread $\alpha$ and is not observable. 
Hence the language $\lang_B$ built with an atomicity assumption on "z=40;chan.send(z)" is indeed a reduction of $\lang_A$.

Let us reason at the level of a Kripke structure.
The goal of such reductions is to structurally detect the following situation in language $\lang$: let $r=a_0^{n_0}a_1^{n_1}a_2^{n_2}\ldots$ designate a run (not necessarily in the language). Then there must exist two indexes $i$ and $j$ such that for any natural number $k$, $i \leq k \leq j$, $r_k=a_0^{n_0}\ldots a_i^{n_i}\ldots a_k^{n_k + 1} \ldots a_j^{n_j} \ldots$ is in the language. In other words, the set of runs $\{ r_k=a_0^{n_0}\ldots a_i^{n_i}\ldots a_k^{n_k + 1} \ldots a_j^{n_j} \ldots \mid i \leq k \leq j\}$ must be included in the language. 
This corresponds to an event that does not impact the truth value of atomic propositions (it stutters) and can be freely commuted with any event that occurs between indexes $i$ and $j$ in the run. This event is simply constrained to occur at the earliest at index $i$ in the run and at the latest at index $j$. In Fig.~\ref{fig:example} the event "z=40" can happen as early as in the initial state, and must occur before "chan.send(z)" and thus matches this definition.

Note that these runs are all stutter equivalent, but are incomparable by the shorter than relation (e.g. $aabc\sigma, abbc\sigma, abcc\sigma$ are incomparable). In this situation, a \emph{reduction} can choose to only represent the run $r$ (e.g. $r=abc\sigma$ represents all these runs) instead of any of these runs. This run was not originally in the language in general, but it is indeed shorter than any of the $r_k$ runs so it matches definition~\ref{def:reduction} for a reduction. Note that the stutter-equivalent class $\hat{r}$ of $r$ does contain all these longer runs so that in a stutter insensitive context, examining $r$ is enough to conclude for any of the runs in $\hat{r}$. This is why usage of structural reductions is compatible with verification of a logic such as $LTL_{\setminus X}$ and has been proposed for that express purpose in the literature~\cite{PPP00,HPP06,Laarman18}.

Thus transaction reductions~\cite{Lipton75,Laarman18} as well as  
both pre-agglomeration and post-agglomeration of Petri nets~\cite{PPP00,EHPP05,HPP06,YTM20} produce a system whose language is a reduction of the language of the original system. 

A formal definition involves a) introducing the syntax of a formalism and b) its semantics in terms of language, then c) defining the reduction rule, and d) proving its effect  is a reduction at the language level. The exercise is not particularly difficult, and the definition of reduction rules mostly fall into the category above, where a non observable event that happens at the earliest at point $i$ in the run and at the latest at point $j$ in the run is abstracted from the trace.
As an example, we provide in Section~\ref{sec:petri} such a proof for structural agglomerations rules on Petri nets.

Our experimental Section~\ref{ss:modelchecking} uses slightly extended versions of these rules (defined in~\cite{YTM20}) for (potentially partial) pre and post-agglomeration.
That paper presents $22$ structural reductions rules from which we selected the rules valid in the context of LTL verification. Only one rule preserving stutter insensitive LTL was not compatible with our approach since it does not produce a reduction at the language level: rule $3$ ``Redundant transitions'' proposes that if two transitions $t_1$ and $t_2$ have the same combined effect as a transition $t$, and firing $t_1$ enables $t_2$, $t$ can be discarded from the net. This reduces the number of edges in the underlying \KS{} representing the state space, but does not affect reachability of states. However, it selects as representative a run involving both $t_1$ and $t_2$ that is longer than the one using $t$ in the original net, it is thus not legitimate to use it in our strategy (although it remains valid for $LTL_{\setminus X}$). 
Rules $14$ ``Pre agglomeration'' and $15$ ``Post agglomeration'' are the most powerful rules of~\cite{YTM20} that we are able to apply in our context. They are known to preserve $LTL_{\setminus X}$ (but not full LTL) and their effect is a reduction at the language level, hence we \emph{can} use them when dealing with shortening/lengthening insensitive formulae. 

\subsection{Petri net agglomeration rules}
\label{sec:petri}

This section presents agglomeration rules for Petri nets, and shows that this transformation of the 
structure of a net leads to a reduction of the associated language in the sense of Definition~\ref{def:reduction}.
All of this section is new content with respect to~\cite{forte22}.

\begin{defi}
[Structure] A Petri net $N=\tuple{\Places,\Trans,\Pre,\Post, m_0}$ is a tuple where \Places{} is the finite set of places, \Trans{} is the finite set of  transitions, $\Pre: \Places \times \Trans \rightarrow \nat$ and $\Post: \Places \times \Trans \rightarrow \nat$ represent the pre and post incidence matrices, and $m_0: \Places \rightarrow \nat$ is the initial marking.
\end{defi}

\textbf{Notations.} 
We use $p$ (resp. $t$) to designate a place (resp. transition) or its index dependent on the context. We represent markings $m$ as vectors of natural numbers with $|\Places|$ entries. We also describe $\Pre(t)$ and $\Post(t)$, for any given transition $t$, as vectors with $|\Places|$ entries.
In vector spaces, we use $v \geq v'$ to denote $\forall i, v(i)\geq v'(i)$, and we use sum $v+v'$ with the usual element-wise definition.
Since a Petri net is bipartite graph, we note $\Preset{n}$ (resp. \Postset{n}) the pre set (resp. post set) of a node $n$ (place or transition). For example, the pre set of place $p$ is $\Preset{p}= \{ t \in \Trans{} \mid \Post(p,t) > 0 \}$.  

\begin{defi}
[Semantics] The semantics of a Petri net is given by the firing rule $\xrightarrow{t}$ that relates pairs of markings: in any marking $m \in  \nat^{|\Places|}$, if $t \in \Trans$ satisfies $m \geq \Pre(t)$, then $m\xrightarrow{t}m'$ with $m' = m + \Post(t) - \Pre(t)$. 
Given a set of atomic propositions $\AP$ and an evaluation function 
$\lambda:  \nat^{|\Places|} \rightarrow 2^\AP$, we can associate to a net $N=\tuple{\Places,\Trans,\Pre,\Post, m_0}$ the corresponding Kripke structure $\KS_\AP=\tuple{\nat^{|\Places|},\{(m,m') \mid \exists t \in \Trans, m\xrightarrow{t}m'\},\lambda,m_0}$ over $\AP$.
\end{defi}

An atomic proposition is a Boolean formula (using $\lor, \land, \lnot$) built from comparisons ($\bowtie \in \{<, \leq, =, \geq, >\}$) of arbitrary weighted sum of place markings to another sum or a constant, e.g. $\sum_{p \in \Places} \alpha_p \cdot m(p) \bowtie k$, with $\alpha_p \in \zrel$ and $k \in \zrel$.
The \textit{support} of a property expressed over $\AP$ is the set of places whose marking is truly used in an atomic predicate, i.e. such that at least one comparison atom has a non zero $\alpha_p$ in a sum. The support $\mathit{Supp} \subseteq \Places{}$ of the property defines the subset $\Stutter{} \subseteq \Trans$ of \textit{invisible} or \textit{stuttering} transitions $t$ satisfying $\forall p \in \mathit{Supp}, \Pre(p,t)=\Post(p,t)$. 
Hence, given a labeling function $\lambda: \nat^\Places{} \rightarrow 2^\AP$, we can guarantee that $\forall t \in \Stutter{}$ and any $m,m' \in \nat^\Places{}$, if $m\xrightarrow{t}m'$ then $\lambda(m)=\lambda(m')$.
With respect to stutter, the atomic propositions $\AP$ are only interested in the projection of reachable markings over the variables in the support, values of places in $\Places{}\setminus \mathit{Supp}$ are not \emph{observable} in markings. A small support means more potential reductions, as rules  mostly cannot apply to observed places or their neighborhood. 

\medskip

Agglomeration in a place $p$ consists in discarding $p$ and its surrounding transitions (pre and post set of $p$) to build instead a transition for every element in the Cartesian product $\Preset{p} \times \Postset{p}$ that represents the effect of the sequence of firing a transition in $\Preset{p}$ then immediately a transition in $\Postset{p}$. 
This ``acceleration'' of tokens in $p$ reduces interleaving in the state space, but can preserve properties of interest if $p$ is chosen correctly. This type of reduction has been heavily studied~\cite{HPP06,Laarman18} as it forms a common ground between structural reductions, partial order reductions and techniques that stem from transaction reduction. 

We present here two rules that can be used to safely decide if a place can be agglomerated.
Figure~\ref{fig:agglo} shows these rules graphically, using the classical graphical notation for Petri nets.

\begin{defi}
[Pre and Post Agglomeration]
Let $N$ designate a Petri net and $\mathit{Supp} \subseteq \Places{}$ be the support of a property on this net.
Let $p \in \Places$ be a place of $N$ satisfying:

\begin{tabular}{@{}ll@{}}
$p \in \Places\setminus \mathit{Supp}$  & $p$ not in support \\
$m_0(p)=0$ & initially unmarked \\
$\Preset{p} \cap \Postset{p} = \emptyset$ & distinct feeders and consumers \\
$\forall h \in \Preset{p}, \Pre(p,h)=1$ &  feeders produce a single token in $p$ \\ 
$\forall f \in \Postset{p}, \Post(p,f)=1$ &  consumers require a single token in $p$ \\ 
\end{tabular}

Pre agglomeration in $p$ is possible if it also the case that:

\begin{tabular}{ll}
$\forall h \in \Preset{p},$ &  \\ 
& $\left\{
\begin{array}{@{}ll@{}}
h \in \Stutter{} & $feeders are stuttering $\\
\Postset{h}=\{p\} & p$ is the single output of $h \\
\exists p_1 \in \Preset{h}, \Post(p_1,h) < \Pre(p_1,h) & h$ is divergent free$ \\
\forall p_2 \in \Preset{h}, \Postset{p_2}=\{h\} & h$ is strongly quasi-persistent$ \\
\end{array}\right. $ \\
\end{tabular}

Post agglomeration in $p$ is possible if it also the case that:

\begin{tabular}{ll}
$\forall f \in \Postset{p},$  \\
& 
$\left\{
\begin{array}{@{}ll@{}}
f \in \Stutter{} & $all consumers are stuttering $\\
\Preset{f} = \{p\} & $no other inputs to $f$ $ \\
\end{array}\right. $ \\
\end{tabular}

In either of these cases, we agglomerate place $p$ leading to a net where
\begin{itemize}
\item we discard all transitions in $\Preset{p} \cup \Postset{p}$ and place $p$ from the net, and 
\item $\forall h \in \Preset{p}, \forall f \in \Postset{p}, $ we add a new transition $t$ such that $\Pre(t)=\Pre(h) + \Pre(f)$ and $\Post(t)=\Post(h) + \Post(f)$
\end{itemize}
\end{defi}

Pre agglomeration (see Fig.~\ref{fig:preagglo}) looks for a place $p$ such that all predecessor transitions $h \in \Preset{p}$
are stuttering, have $p$ as single output place, and once enabled cannot be disabled by firing
any other transition (hence they commute freely with other transitions that do not consume in $p$).
With these conditions, we can always ``delay'' the firing of $h$ until it becomes relevant to
enable a transition that consumes from $p$. This will yield a Petri net whose language is a reduction
of the original one.

Conversely, Post agglomeration (see Fig.~\ref{fig:postagglo}) looks for a place $p$ such that all successor transitions $f \in \Postset{p}$
are stuttering, and have $p$ as single input place.
With these conditions, tokens that arrive in $p$ are always free to choose where they want to go, and nothing can prevent them from going there.
Instead of waiting in $p$ until we take this decision, we can make this choice immediately after firing a transition $h$ that places a token in $p$.
This again will yield a Petri net whose language is a reduction of the original one.

\begin{figure}[tbp]
    \centering
  \begin{subfigure}[t]{\linewidth}
  \centering
\begin{tikzpicture}

\node[place,
    color=red,
    label=$p_0$] (p0) at (0,0) {};

\node[transition,
    minimum height=12mm,
    minimum width=1.5mm,
    fill=black,
    label=$h$] (h) at (1,0) {};

\draw[-latex,thick] (p0) -- (h);

\node[place,
    color=red,
    label=$p$] (p) at (2,0) {};

\draw[-latex,thick] (h) -- (p);

\node[transition,
    minimum height=12mm,
    minimum width=1.5mm,
    fill=black,
    label=$f_1$] (f1) at (3,1) {};

\node[transition,
    minimum height=12mm,
    minimum width=1.5mm,
    fill=black,
    label=$f_2$] (f2) at (3,-1) {};

\node[place,
    label=$p_1$] (p1) at (2,2) {};

\draw[-latex,thick] (p) -- (f1);
\draw[-latex,thick] (p) -- (f2);
\draw[-latex,thick] (p1) -- (f1);

\node[place,
    label=$p_2$] (p2) at (4,2) {};
\node[place,
    label=$p_3$] (p3) at (4,.5) {};
\node[place,
    label=$p_4$] (p4) at (4,-1) {};

\draw[-latex,thick] (f1) -- (p2);
\draw[-latex,thick] (f1) -- (p3);
\draw[-latex,thick] (f2) -- (p4);

\draw[-latex,dashed] (1,2.5) -- (p1);
\draw[-latex,dashed] (1,2) -- (p1);

\draw[-latex,dashed] (p2) -- (5,2);
\draw[-latex,dashed] (p2) -- (5,2.5);
\draw[-latex,dashed] (p3) -- (5,1);
\draw[-latex,dashed] (p4) -- (5,-1);
\draw[-latex,dashed] (3.3,-1.8) -- (p4);

\draw[-latex,dashed] (-.7,-1) -- (p0);

\draw[-latex,dashed,markedge] (p0) -- (.7,-1);
\draw[-latex,dashed,markedge] (h) -- (1.7,-1);
 
\end{tikzpicture}
 
\begin{tikzpicture}

\node[place,
    color=red,
    label=$p_0$] (p0) at (0,0) {};

\node[transition,
    minimum height=12mm,
    minimum width=1.5mm,
    fill=black,
    label=$h.f_1$] (f1) at (3,1) {};

\node[transition,
    minimum height=12mm,
    minimum width=1.5mm,
    fill=black,
    label=$h.f_2$] (f2) at (3,-1) {};

\node[place,
    label=$p1$] (p1) at (2,2) {};

\draw[-latex,thick] (p0) -- (f1);
\draw[-latex,thick] (p0) -- (f2);
\draw[-latex,thick] (p1) -- (f1);

\node[place,
    label=$p_2$] (p2) at (4,2) {};
\node[place,
    label=$p_3$] (p3) at (4,.5) {};
\node[place,
    label=$p_4$] (p4) at (4,-1) {};

\draw[-latex,thick] (f1) -- (p2);
\draw[-latex,thick] (f1) -- (p3);
\draw[-latex,thick] (f2) -- (p4);

\draw[-latex,dashed] (1,2.5) -- (p1);
\draw[-latex,dashed] (1,2) -- (p1);

\draw[-latex,dashed] (-.7,-1) -- (p0);

\draw[-latex,dashed] (p2) -- (5,2);
\draw[-latex,dashed] (p2) -- (5,2.5);
\draw[-latex,dashed] (p3) -- (5,1);
\draw[-latex,dashed] (p4) -- (5,-1);
\draw[-latex,dashed] (3.3,-1.8) -- (p4);
 
\end{tikzpicture}
 
    \caption{
    Pre-agglomeration. The firing of transition $h$ that is stuttering in a run can be delayed until an $f$ transition is fired. The $h$ transition is a right-mover in transaction reduction terms.
    }
    \label{fig:preagglo}
\end{subfigure}
  \begin{subfigure}[t]{\linewidth}
  \centering
\begin{tikzpicture}

\node[place,
    label=$p_0$] (p0) at (0,1) {};

\node[place,
    label=$p_1$] (p1) at (0,-0.5) {};

\node[transition,
    minimum height=12mm,
    minimum width=1.5mm,
    fill=black,
    label=$h$] (h) at (1,0) {};

\draw[-latex,thick] (p0) -- (h);
\draw[-latex,thick] (p1) -- (h);

\node[place,
    color=red,
    label=$p$] (p) at (2,0) {};

\node[place,
    label=$p_2$] (p2) at (2,2) {};

\draw[-latex,thick] (h) -- (p);
\draw[-latex,thick] (h) -- (p2);

\draw[-latex,dashed] (1,2.5) -- (p2);
\draw[-latex,dashed] (p2) -- (3,2.5);
\draw[-latex,dashed] (1,2) -- (p2);

\node[transition,
    minimum height=12mm,
    minimum width=1.5mm,
    fill=black,
    label=$f_1$] (f1) at (3,1) {};

\node[transition,
    minimum height=12mm,
    minimum width=1.5mm,
    fill=black,
    label=$f_2$] (f2) at (3,-1) {};

\draw[-latex,thick] (p) -- (f1);
\draw[-latex,thick] (p) -- (f2);

\node[place,
    color=red,
    label=$p_3$] (p3) at (4,2) {};
\node[place,
    color=red,
    label=$p_4$] (p4) at (4,.5) {};
\node[place,
    color=red,
    label=$p_5$] (p5) at (4,-1) {};

\draw[-latex,thick] (f1) -- (p3);
\draw[-latex,thick] (f1) -- (p4);
\draw[-latex,thick] (f2) -- (p5);

\draw[-latex,dashed] (p3) -- (5,2);
\draw[-latex,dashed] (p3) -- (5,2.5);
\draw[-latex,dashed] (p4) -- (5,1);
\draw[-latex,dashed] (p5) -- (5,-1);
\draw[-latex,dashed] (3.3,-1.8) -- (p5);

\draw[-latex,dashed] (-.7,1) -- (p0);
\draw[-latex,dashed] (-.7,-1) -- (p1);

\draw[-latex,dashed,markedge] (2,1.2) -- (f1);
\draw[-latex,dashed,markedge] (2,-1) -- (f2);
 
\end{tikzpicture}
 
\begin{tikzpicture}

\node[place,
    label=$p_0$] (p0) at (0,1) {};

\node[place,
    label=$p_1$] (p1) at (0,-0.5) {};

\node[transition,
    minimum height=12mm,
    minimum width=1.5mm,
    fill=black,
    label=$h.f_1$] (f1) at (1,0.75) {};

\node[transition,
    minimum height=12mm,
    minimum width=1.5mm,
    fill=black,
    label=$h.f_2$] (f2) at (1,-1) {};

\draw[-latex,thick] (p0) -- (f1);
\draw[-latex,thick] (p1) -- (f1);

\draw[-latex,thick] (p0) -- (f2);
\draw[-latex,thick] (p1) -- (f2);

\node[place,
    label=$p_2$] (p2) at (2,2) {};

\draw[-latex,thick] (f1) -- (p2);

\draw[-latex,thick] (f2) -- (p2);

\draw[-latex,dashed] (1,2.5) -- (p2);
\draw[-latex,dashed] (p2) -- (3,2.5);
\draw[-latex,dashed] (1,2) -- (p2);

\node[place,
    color=red,
    label=$p_3$] (p3) at (4,2) {};
\node[place,
    color=red,
    label=$p_4$] (p4) at (4,.5) {};
\node[place,
    color=red,
    label=$p_5$] (p5) at (4,-1) {};

\draw[-latex,thick] (f1) -- (p3);
\draw[-latex,thick] (f1) -- (p4);
\draw[-latex,thick] (f2) -- (p5);

\draw[-latex,dashed] (p3) -- (5,2);
\draw[-latex,dashed] (p3) -- (5,2.5);
\draw[-latex,dashed] (p4) -- (5,1);
\draw[-latex,dashed] (p5) -- (5,-1);
\draw[-latex,dashed] (3.3,-1.8) -- (p5);

\draw[-latex,dashed] (-.7,1) -- (p0);
\draw[-latex,dashed] (-.7,-1) -- (p1);
 
\end{tikzpicture}
 
    \caption{
    Post-agglomeration. The firing of transition $f$ that is stuttering in a run can be advanced to immediately after the  $h$ transition is fired. The $f$ transition is a left-mover in transaction reduction terms.
    }
    \label{fig:postagglo}
\end{subfigure}
\caption{Pre and Post agglomeration rules, graphically represented. Figures to the left are before agglomeration and to the right after the transformation. Places in red cannot be part of the support of the property.}
\label{fig:agglo}
\end{figure}

\begin{pty}
Let $N$ be a Petri net, and $p$ a place of this net satisfying one of the agglomeration criterion. 
The agglomeration in $p$ of $N$ produces a net whose language is a reduction of the language of $N$.
\end{pty}

\begin{proof} \textit{(Sketch)}
Consider an execution of $N$ of the form:
$m_0 \xrightarrow{t_0} \ldots m_i \xrightarrow{h} m_{i+1} \xrightarrow{t_{i+1}} 
\ldots m_j \xrightarrow{f} m_{j+1} \ldots$
where $h \in \Preset{p}$, $f \in \Postset{p}$ and $\forall k \in \nat, i < k < j$, 
$t_k \not\in \Preset{p} \cup \Postset{p}$.
For pre-agglomeration of $N$ in $p$, the image of this execution in the agglomerated net is of the form:
$m_0 \xrightarrow{t_0} \ldots m_i \xrightarrow{t_{i+1}} m_{i+1}' \ldots m_j' \xrightarrow{h.f } m_{j+1} \ldots$ where the firing of $h$ has shifted to the right until it is fused with the subsequent firing of $f$.
Indeed, we have the strong diamond property that states commutativity of $h$ with any of the transitions between $h$ and $f$ in the execution. Let $t$ be such a transition, if $\exists m,m1,m' \in \nat^\Places{}$ such that  $m \xrightarrow{h} m_1 \xrightarrow{t} m'$ then $\exists m_2, m \xrightarrow{t} m_2 \xrightarrow{h} m'$.
Indeed $h$ cannot be disabled by any transition once it is enabled (due to ``strongly quasi persistent'' property), and firing of $h$ cannot enable any transition not in $\Postset{p}$ (since $\Postset{h}=\{p\}$). 
The semantics of a Petri net also induce that firing $h$ before or after $t$ lead to the same marking.
Furthermore, the underlying run in the agglomerated net is shorter (by one repetition) than the original run, because $h$ is guaranteed to stutter (it is one of the constraints that $h \in \Stutter{}$).

If the execution of the original net contains a firing of transition $h$ but not of $f$, 
$m_0 \xrightarrow{t_0} \ldots m_i \xrightarrow{h} m_{i+1} \ldots$
where $h \in \Preset{p}$, and $\forall k \in \nat, k > i$, $t_k \not\in \Postset{p}$, the image is the execution
$m_0 \xrightarrow{t_0} \ldots m_i \xrightarrow{t_{i+1}} m_{i+1}' \ldots$ where the transition $h$ is never fired.
Since $h$ is stuttering, this corresponding run is shorter than the original (by one repetition).

If the execution does not contain any firing of a transition in $\Preset{p} \cup \Postset{p}$, the image of the execution is the execution itself which is legitimate since identity is a reduction.
Lastly, the divergent free constraint ensures that $h$ cannot be fired infinitely many times in succession so that
the original net cannot exhibit an execution of the form $m_0 \ldots m' \xrightarrow{h} m'' \xrightarrow{h}\ldots$
(ending on infinitely many firings of $h$) that would no longer exist in the agglomerated net.

A similar reasoning is possible for post agglomeration, given an execution 
$m_0 \xrightarrow{t_0} \ldots m_i \xrightarrow{h} m_{i+1} \xrightarrow{t_{i+1}} 
\ldots m_j \xrightarrow{f} m_{j+1} \ldots$
where $h \in \Preset{p}$, $f \in \Postset{p}$ and $\forall k \in \nat, i < k < j$, 
$t_k \not\in \Preset{p} \cup \Postset{p}$ of the original net, the image in the agglomerated net is 
$m_0 \xrightarrow{t_0} \ldots m_i \xrightarrow{h.f} m_{i+1}' \xrightarrow{t_{i+1}}\ldots$ where the firing of $f$ has shifted to the left until it is fused with the preceding firing of $h$.
Again the diamond property applies, $f$ is enabled as soon as $h$ fires and cannot subsequently be disabled by other transitions.
Because $f$ is constrained to be stuttering, the corresponding run is shorter than the original one.
\end{proof}

The pre and post agglomeration rules presented here are variants on the rules initially introduced by Berthelot~\cite{Berthelot85}, and have been used in the context of stutter insensitive LTL model-checking~\cite{PPP00,HPP06}.

\section{Experimentation}
\label{sec:perfs}

\subsection{A Study of Properties}
\label{sec:perfsLTL}

This section provides an empirical study of the applicability of the techniques presented in this paper to LTL properties found in the literature. To achieve this we explored several LTL benchmarks~\cite{etessami.00.concur,somenzi.00.cav,dwyer.98.fmsp,rers21,mcc:2021}.
Some work~\cite{etessami.00.concur,somenzi.00.cav} summarizes  the typical properties that users express in LTL. The formulae of this benchmark have been extracted  directly from the literature. 
Dwyer et al.~\cite{dwyer.98.fmsp}  propose property specification patterns, expressed in several logics including LTL. These patterns have been extracted by analysing 447 formulae coming from real world projects.
The RERS challenge~\cite{rers21} presents generated formulae inspired from real world reactive systems. 
    The MCC~\cite{mcc:2021} benchmark establishes a huge database of $45152$ LTL formulae in the form of $1411$ 
    Petri net models coming from $114$ sources with $32$ random LTL formulae for each one. These formulae use up to $5$ state-based atomic propositions (e.g. "$m(p0)+m(p1)>2$" or "$t1$ is fireable"), limit the nesting depth of temporal operators to $5$ and are filtered by the organizers in order to be non trivial. Since these formulae come with a concrete system we were able to use this benchmark to also provide performance results for our approach in Section~\ref{ss:modelchecking}. We retained $43989$ model/formula pairs from this benchmark; the missing $1163$ were rejected due to parse limitations of our tool when the model size is excessive ($>10^7$ transitions).
This set of roughly $2200$ formulae exhibiting real patterns and $44,000$ random ones lets us evaluate if the fragment of LTL that we consider is common in practice. 
Table~\ref{tab:formulae} summarizes, for each benchmark, the number and percentage of formulae that are either stuttering insensitive, lengthening insensitive, or shortening insensitive. 
The sum of both shortening and lengthening  formulae represents more than one third (and up to 60 percent) of the formulae of these benchmarks. 

Concerning the polarity, although lengthening insensitive formulae seem to appear more frequently, 
most of these benchmarks actually contain each formula in both positive and negative forms (we retained only one) so that the summed percentage might be more relevant as a metric since lengthening insensitivity of $\varphi$ is equivalent to shortening insensitivity of $\lnot\varphi$.
Analysis of the human generated Dwyer patterns~\cite{dwyer.98.fmsp} reveals that shortening/lengthening insensitive formulae mostly come from the patterns \textit{precedence chain}, \textit{response chain} and \textit{constrained chain}.
These properties specify causal relation between events, which are observable as causal relations 
between \emph{observably different} states (that might be required to strictly follow each other), but this causality chain is not impacted by non observable events.


\begin{figure}

    \begin{tabular}{l||r|r|r|r|r}
        Benchmark      & Total & SI & LI & ShI & LS \\
\hline
         Dwyer et al.~\cite{dwyer.98.fmsp} & 55 & 32 (58\%)     & 13  (24\%) & 9 (16\%)   & 1 (2\%)     \\
         Spot~\cite{etessami.00.concur,dwyer.98.fmsp,somenzi.00.cav} & 94  & 63 (67\%)    & 17 (18\%)   & 11 (12\%)   & 3 (3\%)    \\
         RERS~\cite{rers21}. & 2050 & 714 (35\%)  & 777  (38\%) & 559 (27\%) & 0 \\
         MCC~\cite{mcc:2021} & 43989 & 24462 (56\%) & 6837 (16\%) &5390 (12\%) & 7300 (17\%)  \\
    \end{tabular}
    \captionof{table}{\small{Sensitivity to length of properties measured using several LTL benchmarks. We distinguish stutter insensitive (SI), lengthening insensitive (LI), shortening insensitive (ShI), and length sensitive (LS) formulae.}}
    \label{tab:formulae}
\end{figure}
\begin{figure}
  \begin{subfigure}[t]{0.45\textwidth}
    \includegraphics[width=\textwidth]{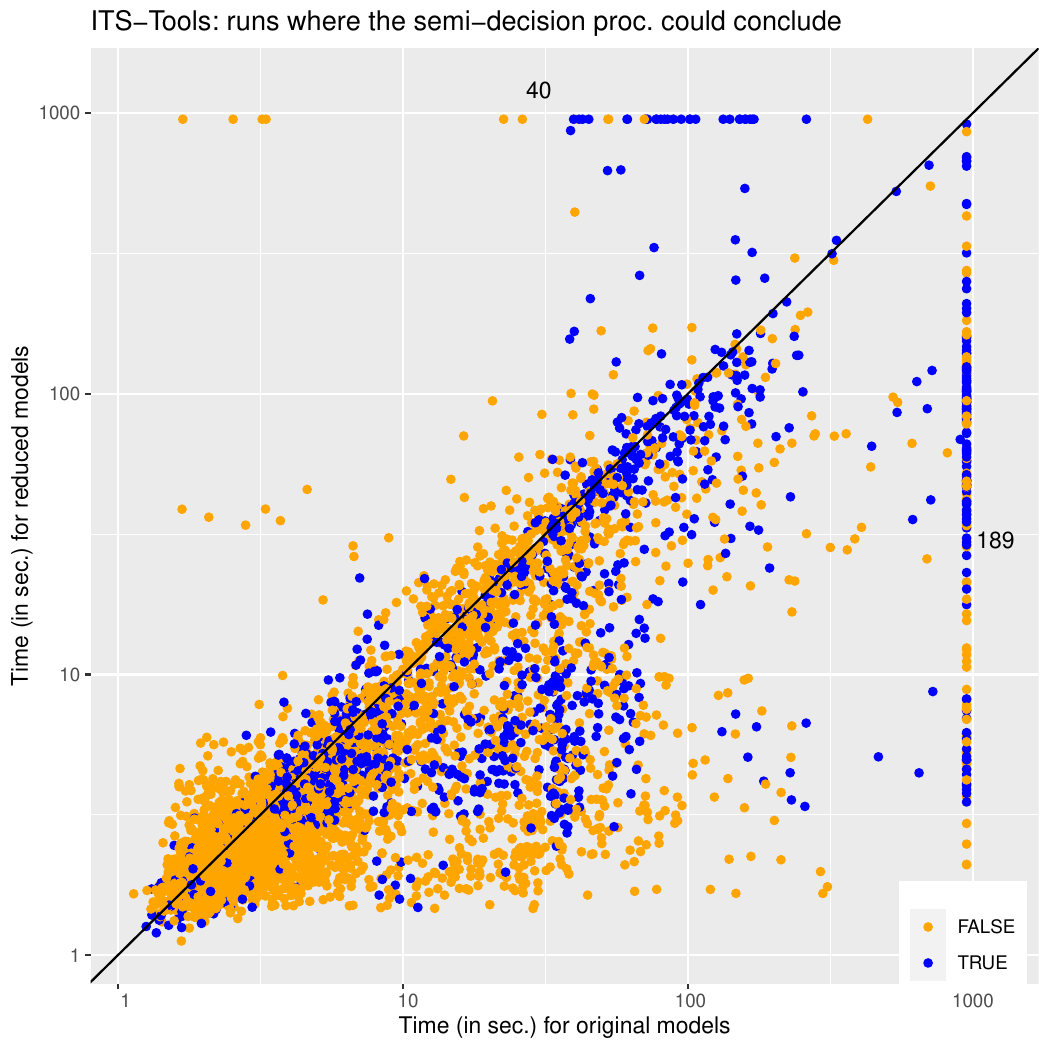}
    \caption{ITS-tools on decidable instances}
    \label{fig:1}
  \end{subfigure}
 \hfill
  \begin{subfigure}[t]{0.45\textwidth}
    \includegraphics[width=\textwidth]{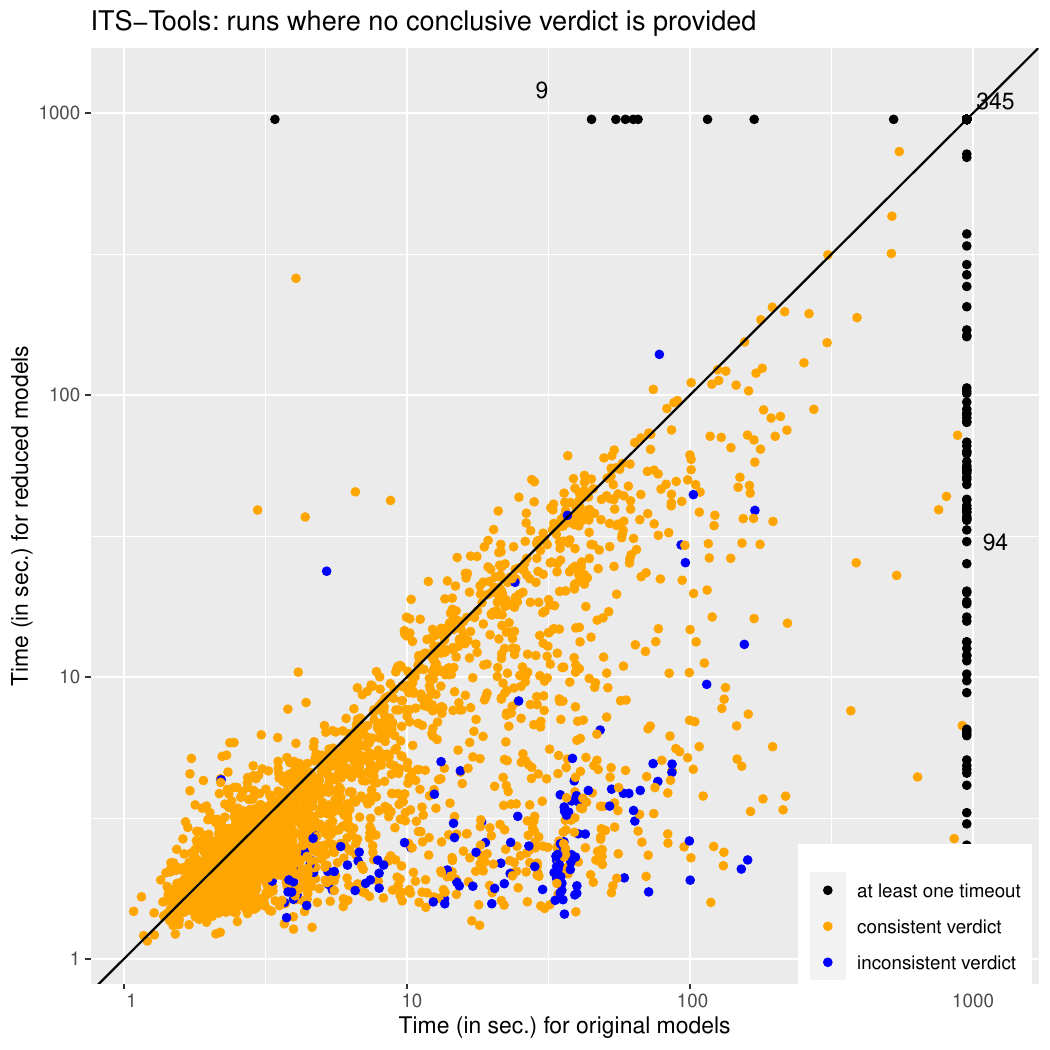}
    \caption{ITS-tools on non decidable instances}
    \label{fig:2}
  \end{subfigure}
  \begin{subfigure}[t]{0.45\textwidth}
    \includegraphics[width=\textwidth]{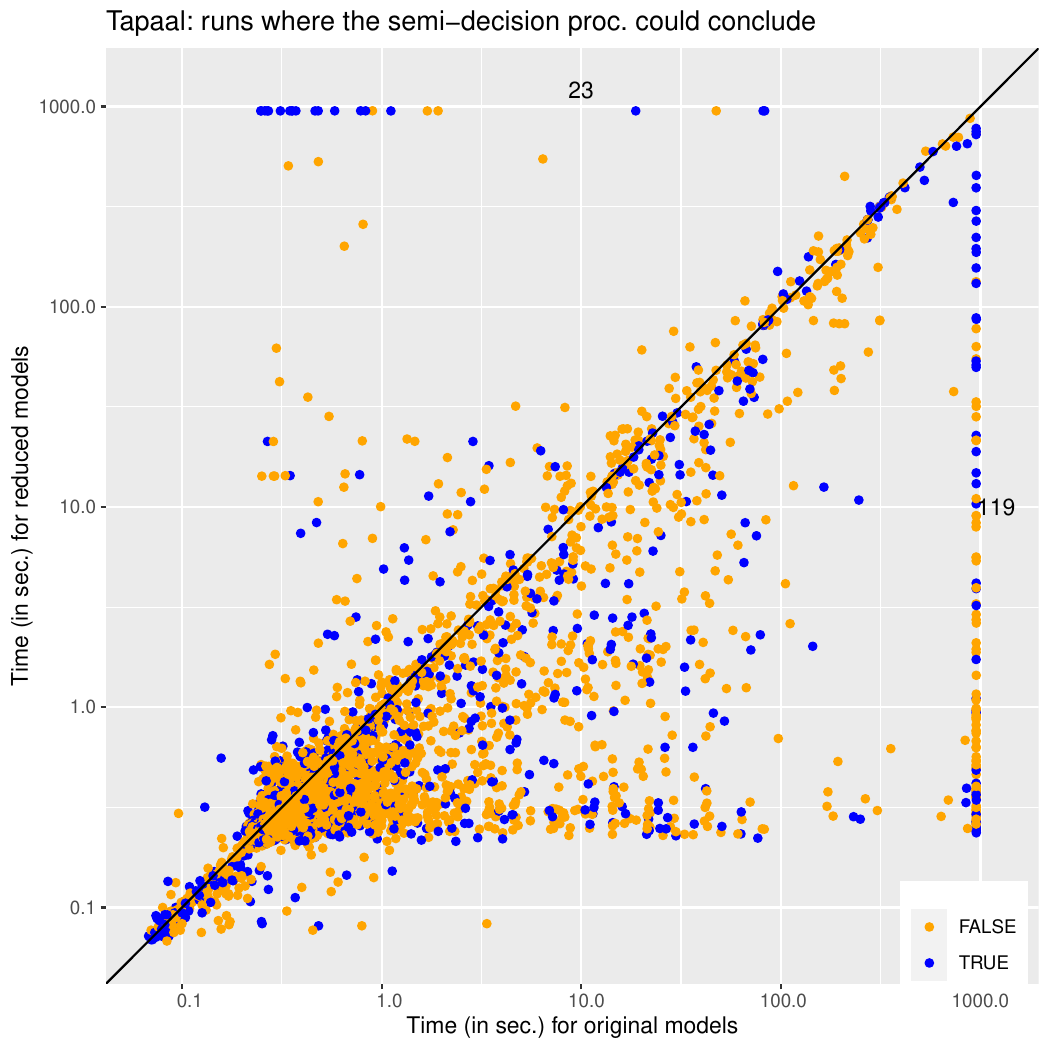}
    \caption{Tapaal on decidable instances}
    \label{fig:3}
  \end{subfigure}
 \hfill
  \begin{subfigure}[t]{0.45\textwidth}
    \includegraphics[width=\textwidth]{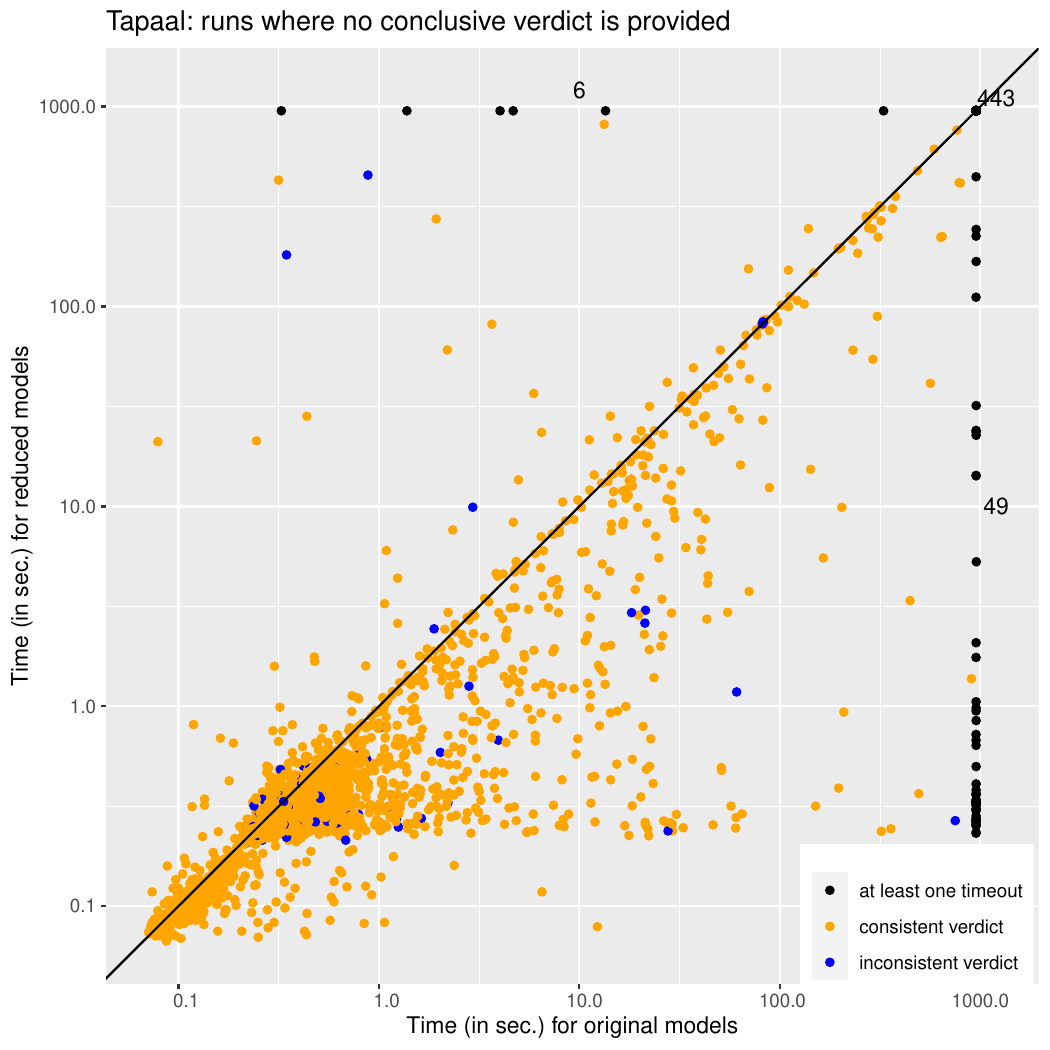}
    \caption{Tapaal on non decidable instances}
    \label{fig:4}
  \end{subfigure}
  \caption{\small{Experiments on the MCC'2021 LTL benchmark using the two best tool of the MCC contest: Tapaal and ITS-tools. These plots show the results on the $9222$ examples (21\% of the original $43989$ examples of the MCC) where the formula is shortening or lengthening insensitive (but not both) and the system can be reduced.
Figures (a) and (c) contain the cases where the verdict of the semi-decisions procedures is reliable, 
and distinguish cases where the output is True (empty product) and False (non empty product). 
(b) and (d) display the cases where the verdict is not reliable and distinguish cases where the output is
inconsistent with the ground truth from cases where they agree. }
}
\label{fig:benchmark}

\end{figure}

\subsection{A Study of Performances}
\label{ss:modelchecking}

\textbf{Benchmark Setup.}
Among the LTL benchmarks presented in Table~\ref{tab:formulae}, we opted for the MCC benchmark 
to evaluate the techniques presented in this paper. 
This benchmark seems relevant since 
(1) it contains both academic and industrial models, 
(2) it has a huge set of (random) formulae and
(3)  includes models so that we could measure the effect of the approach in a model-checking setting. 
The model-checking competition (MCC) is an annual event in its $10^{th}$ edition in 2021 where competing tools 
are evaluated on a large benchmark. 
We use the formulae and models from the latest $2021$ edition of the contest, where Tapaal~\cite{tapaal12} was awarded the gold medal
and ITS-Tools~ \cite{ITStools} was silver in the LTL category of the contest.
We evaluate both of these tools in the following performance measures, showing that our strategy is agnostic
to the back-end analysis engine. Our experimental setup consists in two steps.
\begin{enumerate}
    \item Parse the model and formula pair, and analyze the sensitivity of the formula. When the formula is 
    shortening or lengthening insensitive (but not both) output two model/formula pairs: reduced and original.
    The ``original'' version does also benefit from structural simplification rules (e.g. removing redundant places and transitions, implicit places that are always sufficiently marked to fire related transitions\ldots), but we apply only rules that are compatible with full LTL (in particular not enabling agglomeration rules).
    The ``reduced'' version additionally benefits from rules that are reductions at the language level, in the sense of Definition~\ref{def:reduction}. These rules mainly include pre and post agglomeration, that also can enable even more structural simplification rules, since we iterate all applicable reduction rules to a fixed point. 
    The original and reduced model/formula pairs that result from this procedure are then exported in the same format the contest uses. This step was implemented within ITS-tools.
    \item Run an MCC compatible tool on both the reduced and original versions of each model/formula pair and record the time performance and the verdict. 
\end{enumerate}

For the first step, using Spot~\cite{ADL15}, we detect that a formula is either shortening or lengthening insensitive for 99.81\% of formulae in less than 1 second. After this analysis, we obtain 12\,227 model/formula pairs where the formula is either shortening insensitive or lengthening insensitive (but not both). Among these pairs, in $3005$ cases (24.6\%) none of the structural reduction rules we use were applicable. For more details on the precise structural reductions, see the companion artifact paper~\cite{scico24}.
Since our strategy does not improve such cases, we retain the remaining 9\,222 (75.4\%) 
model/formula pairs in the performance plots of Figure~\ref{fig:benchmark}. 
We measured that on average 34.19\% of the places and 32.69\% of the transitions of the models were discarded by reduction rules with respect to the ``original'' model, though
the spread is high as there are models that are almost fully reducible and some that are barely so.
Application of reduction rules is in complexity related to the size of the structure of the net and 
takes less than 20 seconds to compute in 95.5\% of the models.
We are able to treat $9222$ examples (21\% of the original $43989$ model/formula pairs of the MCC) using reductions. All these formulae until now could not be handled using reduction techniques.

For the second step, we measured the solution time for both reduced and original model/formula pairs using the two best tools of the MCC'2021 contest. 
A full tool using our strategy might optimistically first run on the reduced model/formula pair hoping for a definitive answer, but we recommend the use of a portfolio approach where the first definitive answer is kept.
In these experiments we neutrally measured the time for taking a semi-decision on the reduced model vs. the time for taking a (complete) decision on the original model. We then classify the results into two sets, decidable instances are shown on the left of Fig.~\ref{fig:benchmark} and instances that are not decidable (by our procedure) are on the right. 
On ``decidable instances'' our semi-decision procedure could have concluded reliably because the system satisfies the formula and the property is shortening insensitive, or the system does not satisfy the formula and the property is lengthening insensitive. 
Non decidable instances shown on the right are those where the verdict on the reduced model is not to be trusted
(or both the original and reduced procedures timed out).

With this workflow we show that our approach is generic and can be easily implemented on top of any MCC compatible model-checking tool. All experiments were run with a 950 seconds timeout (close to 15 minutes, which is generous when the contest offers 1 hour for 16 properties).
We used a heterogeneous cluster of machines with four cores allocated to each experiment, and ensured that experiments concerning reduced and original versions of a given model/formula 
are comparable (by running them on the same architecture). To help reproducibility of these experiments, a software artifact is presented in~\cite{scico24}, and available from the CodeOcean platform \url{https://doi.org/10.24433/CO.6846969.v1}.

Figure~\ref{fig:benchmark} presents the results of these  experiments.
The results are all presented as log-log scatter plots opposing a run on the original to a run on the reduced model/formula pair.
Each dot represents an experiment on a model/formula pair; a dot below the diagonal indicates that the reduced version was faster to solve, while a point above it indicates a case where the reduced model actually took longer to solve than the original (fortunately there are relatively few)\footnote{An example found in the benchmark is a formula requiring $X X X G p$, original system has no states satisfying $p$ 3 steps from initial state, but the reduced model does; and it takes a very long sequence of steps to reinvalidate $p$}. 
Points that timeout for one (or both) of the approaches are plotted on the line at 950 seconds, we also indicate the number of points that are in this line (or corner) next to it.

The plots on the left  (a) and (c) correspond to ``decidable instances'' while those on the right are not decidable by our procedure.
The two plots on the top correspond to the performance of ITS-tools, while those on the bottom give the results with Tapaal.  
The general form of the results with both tools is quite similar confirming that our strategy is indeed responsible for the measured gains in performance and that they are reproducible. Reduced problems \emph{are}  generally easier to solve than the original.
This gain is in the best case exponential as is visible through the existence of spread of points reaching out horizontally in this log-log space (particularly on the Tapaal plots).

The colors on the decidable instances reflect whether the verdict was true or false.
For false properties a counter-example was found by both procedures interrupting the search, and while the search space of
a reduced model is a priori smaller, heuristics and even luck can play a role in finding a counter-example early.
True answers on the other hand generally require a full exploration of the state space so that the reductions should play a major role in reducing the complexity of model-checking. The existence of True answers where the reduction fails is surprising at first, but a smaller Kripke structure does not necessarily induce a smaller product as happens sometimes in this large benchmark (and in other reduction techniques such as stubborn sets~\cite{Valmari90}). 
On the other hand the points aligned to the right of the plots a) and c) (189 for ITS-tools and 119 for Tapaal) 
correspond to cases where our procedure improved these state of the art tools, allowing one to reach a conclusion when the original method fails.

The plots on the right use orange to denote cases where the verdict on the reduced and original models were the same; on these points
the procedures had comparable behaviors (either exploring a whole state space or exhibiting a counter-example). The blue color denotes points where the two procedures disagree, with several blue points above the diagonal reflecting cases where the reduced procedure explored the whole state space and thought the property was true while the original procedure found a counter-example (this is the worst case).
Surprisingly, even though on these non decidable plots b) and d) our procedure should not be trusted, it mostly agrees (in 95\% of the cases) with the decision reached on the original.

Out of the 9222 experiments in total, for ITS-tools 5901 runs reached a trusted decision (64 \%), 2927 instances reached an untrusted verdict (32 \%), and the reduced procedure timed out in 394 instances (4 \%). Tapaal reached a trusted decision in 5866 instances (64 \%), 2884 instances reached an untrusted verdict (31 \%), and the reduced procedure timed out in 472 instances (5 \%). On this benchmark of formulae we thus reached a trusted decision in almost two thirds of the cases using the reduced procedure.

\section{Extensions and Perspectives}
\label{sec:extension}

In this section we investigate how to improve the results obtained in the previous sections,  
thus extending the presentation of~\cite{forte22}.
The main extension considered is to focus on fragments of a language that satisfy the shortening/lengthening/stutter insensitivity criterion, even when the language as a whole is length sensitive.
We first introduce tools to partition a property language into these fragments and present experimental 
evaluation on a benchmark of LTL formulae. 
We then present possible approaches to confirm non trustworthy answers of the current semi decision procedure,
 one of which turns out to be applicable to a much wider range of properties.
All of this section is new content with respect to~\cite{forte22}.

\subsection{Partition of a Language}

The idea is to partition the language of the (negation of the) property into four parts:
the stutter insensitive part, the pure shortening insensitive part, the pure lengthening insensitive part, and the pure fully length sensitive part. This is a partition of the language, hence the ``pure'' qualifier indicates that there is no overlap, e.g. the pure shortening insensitive part does not contain any word of the stutter insensitive part of the original language.

Given this partition, we can use the most efficient available procedure to check emptiness of the product of the system with each part.
For the first three parts in particular, we can  work with a reduction of the system.
The stutter insensitive part benefits from a full decision procedure when working with the reduction, 
and the partly length insensitive parts can use the semi-decision procedure of Theorem~\ref{th:short}.

\begin{defi}
\label{def:part}
Length-sensitive partition of a language.

Given a language $\lang$, we define:

\begin{itemize}
\item The stutter insensitive part of $\lang$ noted $SI^\pm(\lang)$ is defined as the largest subset of $\lang$ that is stutter insensitive. More precisely $\forall r \in \lang$, 
$$
r \in SI^\pm(\lang) \iff (\forall r' \in \Sigma^\omega, r' \pre r \lor r \pre r' \implies r' \in \lang)
$$
\item The pure shortening insensitive part of $\lang$ noted $SI^-(\lang)$ is the largest subset of $\lang$ that is shortening insensitive but does not intersect with $SI^\pm(\lang)$. More precisely $\forall r \in \lang$

$$
r \in SI^-(\lang) \iff 
\left\{
\begin{array}{l}
\forall r' \in \Sigma^\omega, r' \pre r \implies r' \in \lang \\
\exists r'' \in \Sigma^\omega \setminus \lang, r \pre r''
\end{array}
\right.
$$

\item The pure lengthening insensitive part of $\lang$ noted $SI^+(\lang)$ is the largest subset of $\lang$ that is lengthening insensitive but does not intersect with $SI^\pm(\lang)$. More precisely  $\forall r \in \lang$

$$
r \in SI^+(\lang) \iff 
\left\{
\begin{array}{l}
\forall r' \in \Sigma^\omega, r \pre r' \implies r' \in \lang \\
\exists r'' \in \Sigma^\omega \setminus \lang, r'' \pre r
\end{array}
\right.
$$

\item The pure fully length sensitive part of $\lang$ noted $SS(\lang)$ is the largest subset of $\lang$ that is 
 shortening and lengthening sensitive but does not intersect with the previous parts. More precisely  $\forall r \in \lang$

$$
r \in SS(\lang) \iff 
\left\{
\begin{array}{l}
\exists r' \in \Sigma^\omega \setminus \lang, r \pre r' \\
\exists r'' \in \Sigma^\omega \setminus \lang, r'' \pre r
\end{array}
\right.
$$

\end{itemize}
\end{defi}

To compute this partition when the language $\lang$ is represented by a Buchi automaton $A$, we can reuse the closure ($cl$) and self-loopization ($sl$) syntactic transformations of an automaton introduced in Section~\ref{sec:detection}.

\begin{defi}
Computing the length-sensitive partition of a language.

Let $A$ designate a Büchi automaton recognizing language $\lang$.
\begin{itemize}

\item Let $B=sl(cl(A))$, let $C=B \otimes \bar{A}$, let $D=sl(cl(C))$, then $SI^\pm(A) = A \otimes \overline{D}$

\item Let $A'= A \otimes \overline{SI(A)}$, 
\begin{itemize}
    \item Let $B^-=cl(A')$, let $C^-=B^- \otimes \overline{A'}$, let $D^-=sl(C^-)$, then $SI^-(A) = A' \otimes \overline{D^-}$.
    \item Let $B^+=sl(A')$, let $C^+=B^+ \otimes \overline{A'}$, let $D^+=cl(C^+)$, then $SI^+(A) = A' \otimes \overline{D^+}$.
\end{itemize}

\item Finally, $SS(A) = A \otimes \overline{SI(A)} \otimes \overline{SI^-(A)}  \otimes \overline{SI^+(A)} $

\end{itemize}
\end{defi}

The automaton produced in each case matches Definition~\ref{def:part}, e.g. $\lang(SI^\pm(A)) = SI^\pm(\lang)$ when $A$ recognizes language $\lang$.
The reasoning behind these definitions is illustrated in Figure~\ref{fig:compute}.
For each equivalence class, we simply reason on three distinct but comparable words $r, r'$ and $r''$ such that $r \pre r' \pre r''$.
The question for each of these words is whether they belong to the language of $A$ or not.
This gives us $2^3=8$ cases to consider, but some cases are redundant, e.g. the case where $r$ and $r'$ belong to $A$ but $r''$ doesn't is homogeneous to the case where $r$ belongs to $A$ but neither $r'$ nor $r''$ do.
Figure~\ref{fig:compute} thus only represents the $6$ significantly different cases for a given equivalence class of runs.

\begin{figure}[tbp]
    \centering
  \begin{subfigure}[t]{\linewidth}
  \centering
\begin{tikzpicture}[radius=32mm]

\begin{scope}[shift={(0,0)}]
  \foreach \i [count=\ii from 0] in {white, white, lightgray, lightgray, lightgray, white}
    \filldraw[fill=\i] (0,0) -- (\ii*60:32mm)
                    arc[start angle=\ii*60,delta angle=60]
                    -- cycle;

    \begin{scope}[rotate=60]
        \filldraw[fill=lightgray, draw=none] (0,0) -- (2.1,0) arc (1:60:2.1cm) -- cycle;
    \end{scope}

    \begin{scope}[rotate=120]
        \filldraw[fill=white, draw=none] (0,0) -- (2.1,0) arc (1:60:2.1cm) -- cycle;
    \end{scope}

    \begin{scope}[rotate=240]
        \filldraw[fill=white, draw=none] (0,0) -- (2.1,0) arc (1:60:2.1cm) -- cycle;
        \filldraw[fill=lightgray, draw=none] (0,0) -- (1.5,0) arc (1:60:1.5cm) -- cycle;
    \end{scope}

    \begin{scope}[rotate=300]
        \filldraw[fill=lightgray, draw=none] (0,0) -- (2.1,0) arc (1:60:2.1cm) -- cycle;
        \filldraw[fill=white, draw=none] (0,0) -- (1.5,0) arc (1:60:1.5cm) -- cycle;
    \end{scope}

    \draw (0,0) circle[radius=32mm, draw=none];

    \node[inner sep=2pt, transform shape] (AA) at (0,-3.8)  {\huge{$A$}};

        \foreach \i [count=\ii from 0] in {white, white, lightgray, white, lightgray, white}
    \draw (0,0) -- (\ii*60:32mm)
                    arc[start angle=\ii*60,delta angle=60]
                    -- cycle;

\end{scope}

\begin{scope}[shift={(7,0)}]
  \foreach \i [count=\ii from 0] in {white, lightgray, lightgray, lightgray, lightgray,lightgray}
    \filldraw[fill=\i] (0,0) -- (\ii*60:32mm)
                    arc[start angle=\ii*60,delta angle=60]
                    -- cycle;

    \draw (0,0) circle[radius=32mm, draw=none];

    \node[inner sep=2pt, transform shape] (AA) at (0,-3.8)  {\huge{$B = sl(cl(A))$}};
    \foreach \i [count=\ii from 0] in {white, white, lightgray, white, lightgray, white}
    \draw (0,0) -- (\ii*60:32mm)
                    arc[start angle=\ii*60,delta angle=60]
                    -- cycle;

\end{scope}

\begin{scope}[shift={(14,0)}]
  \foreach \i [count=\ii from 0] in {white, lightgray, white, white, white, lightgray}
    \filldraw[fill=\i] (0,0) -- (\ii*60:32mm)
                    arc[start angle=\ii*60,delta angle=60]
                    -- cycle;

    \begin{scope}[rotate=60]
        \filldraw[fill=white, draw=none] (0,0) -- (2.1,0) arc (1:60:2.1cm) -- cycle;
    \end{scope}

    \begin{scope}[rotate=120]
        \filldraw[fill=lightgray, draw=none] (0,0) -- (2.1,0) arc (1:60:2.1cm) -- cycle;
    \end{scope}

    \begin{scope}[rotate=240]
        \filldraw[fill=lightgray, draw=none] (0,0) -- (2.1,0) arc (1:60:2.1cm) -- cycle;
        \filldraw[fill=white, draw=none] (0,0) -- (1.5,0) arc (1:60:1.5cm) -- cycle;
    \end{scope}

    \begin{scope}[rotate=300]
        \filldraw[fill=white, draw=none] (0,0) -- (2.1,0) arc (1:60:2.1cm) -- cycle;
        \filldraw[fill=lightgray, draw=none] (0,0) -- (1.5,0) arc (1:60:1.5cm) -- cycle;
    \end{scope}
    
    \draw (0,0) circle[radius=32mm, draw=none];

    \node[inner sep=2pt, transform shape] (AA) at (0,-3.8)  {\huge{$C=B \otimes \bar{A}$}};

        \foreach \i [count=\ii from 0] in {white, white, lightgray, white, lightgray, white}
    \draw (0,0) -- (\ii*60:32mm)
                    arc[start angle=\ii*60,delta angle=60]
                    -- cycle;

\end{scope}

\begin{scope}[shift={(21,0)}]
  \foreach \i [count=\ii from 0] in {white, lightgray, lightgray, white, lightgray,lightgray}
    \filldraw[fill=\i] (0,0) -- (\ii*60:32mm)
                    arc[start angle=\ii*60,delta angle=60]
                    -- cycle;

    \draw (0,0) circle[radius=32mm, draw=none];

    \node[inner sep=2pt, transform shape] (AA) at (0,-3.8)  {\huge{$D = sl(cl(C))$}};

        \foreach \i [count=\ii from 0] in {white, white, lightgray, white, lightgray, white}
    \draw (0,0) -- (\ii*60:32mm)
                    arc[start angle=\ii*60,delta angle=60]
                    -- cycle;

\end{scope}

\begin{scope}[shift={(28,0)}]
  \foreach \i [count=\ii from 0] in {white, white, white, lightgray, white,white}
    \filldraw[fill=\i] (0,0) -- (\ii*60:32mm)
                    arc[start angle=\ii*60,delta angle=60]
                    -- cycle;

    \draw (0,0) circle[radius=32mm, draw=none];

    \node[inner sep=2pt, transform shape] (AA) at (0,-3.8)  {\huge{$SI^\pm(A) = A \otimes \bar{D}$}};

        \foreach \i [count=\ii from 0] in {white, white, lightgray, white, lightgray, white}
    \draw (0,0) -- (\ii*60:32mm)
                    arc[start angle=\ii*60,delta angle=60]
                    -- cycle;

\end{scope}

\end{tikzpicture}
    \caption{
    Computing the stutter insensitive part $SI^\pm$ of the language of an automaton $A$. }
\end{subfigure}
  \begin{subfigure}[t]{\linewidth}
  \centering
\begin{tikzpicture}[radius=32mm]

\begin{scope}[shift={(0,0)}]
  \foreach \i [count=\ii from 0] in {white, white, lightgray, white, lightgray, white}
    \filldraw[fill=\i] (0,0) -- (\ii*60:32mm)
                    arc[start angle=\ii*60,delta angle=60]
                    -- cycle;

    \begin{scope}[rotate=60]
        \filldraw[fill=lightgray, draw=none] (0,0) -- (2.1,0) arc (1:60:2.1cm) -- cycle;
    \end{scope}

    \begin{scope}[rotate=120]
        \filldraw[fill=white, draw=none] (0,0) -- (2.1,0) arc (1:60:2.1cm) -- cycle;
    \end{scope}

    \begin{scope}[rotate=240]
        \filldraw[fill=white, draw=none] (0,0) -- (2.1,0) arc (1:60:2.1cm) -- cycle;
        \filldraw[fill=lightgray, draw=none] (0,0) -- (1.5,0) arc (1:60:1.5cm) -- cycle;
    \end{scope}

    \begin{scope}[rotate=300]
        \filldraw[fill=lightgray, draw=none] (0,0) -- (2.1,0) arc (1:60:2.1cm) -- cycle;
        \filldraw[fill=white, draw=none] (0,0) -- (1.5,0) arc (1:60:1.5cm) -- cycle;
    \end{scope}

    \draw (0,0) circle[radius=32mm, draw=none];

    \node[inner sep=2pt, transform shape] (AA) at (0,-3.8)  {\huge{$A' = A \otimes \overline{SI^\pm(A)}$}};

    \foreach \i [count=\ii from 0] in {white, white, lightgray, white, lightgray, white}
    \draw (0,0) -- (\ii*60:32mm)
                    arc[start angle=\ii*60,delta angle=60]
                    -- cycle;

\end{scope}

\begin{scope}[shift={(7,0)}]
  \foreach \i [count=\ii from 0] in {white, white, lightgray, white, lightgray, white}
    \filldraw[fill=\i] (0,0) -- (\ii*60:32mm)
                    arc[start angle=\ii*60,delta angle=60]
                    -- cycle;

    \begin{scope}[rotate=60]
        \filldraw[fill=lightgray, draw=none] (0,0) -- (2.1,0) arc (1:60:2.1cm) -- cycle;
    \end{scope}

    \begin{scope}[rotate=300]
        \filldraw[fill=lightgray, draw=none] (0,0) -- (2.1,0) arc (1:60:2.1cm) -- cycle;
    \end{scope}

    \draw (0,0) circle[radius=32mm, draw=none];

    \node[inner sep=2pt, transform shape] (AA) at (0,-3.8)  {\huge{$B^-=cl(A')$}};

        \foreach \i [count=\ii from 0] in {white, white, lightgray, white, lightgray, white}
    \draw (0,0) -- (\ii*60:32mm)
                    arc[start angle=\ii*60,delta angle=60]
                    -- cycle;

\end{scope}

\begin{scope}[shift={(14,0)}]
  \foreach \i [count=\ii from 0] in {white, white, white, white, white, white}
    \filldraw[fill=\i] (0,0) -- (\ii*60:32mm)
                    arc[start angle=\ii*60,delta angle=60]
                    -- cycle;

    \begin{scope}[rotate=120]
        \filldraw[fill=lightgray, draw=none] (0,0) -- (2.1,0) arc (1:60:2.1cm) -- cycle;
    \end{scope}

   \begin{scope}[rotate=240]
        \filldraw[fill=lightgray, draw=none] (0,0) -- (2.1,0) arc (1:60:2.1cm) -- cycle;
        \filldraw[fill=white, draw=none] (0,0) -- (1.5,0) arc (1:60:1.5cm) -- cycle;
    \end{scope}

    \begin{scope}[rotate=300]
        \filldraw[fill=lightgray, draw=none] (0,0) -- (1.5,0) arc (1:60:1.5cm) -- cycle;
    \end{scope}

    \draw (0,0) circle[radius=32mm, draw=none];

    \node[inner sep=2pt, transform shape] (AA) at (0,-3.8)  {\huge{$C^-=B^- \otimes \overline{A'}$}};

        \foreach \i [count=\ii from 0] in {white, white, lightgray, white, lightgray, white}
    \draw (0,0) -- (\ii*60:32mm)
                    arc[start angle=\ii*60,delta angle=60]
                    -- cycle;

\end{scope}


\begin{scope}[shift={(21,0)}]
  \foreach \i [count=\ii from 0] in {white, white, lightgray, white, lightgray,lightgray}
    \filldraw[fill=\i] (0,0) -- (\ii*60:32mm)
                    arc[start angle=\ii*60,delta angle=60]
                    -- cycle;
    
    \begin{scope}[rotate=240]
        \filldraw[fill=white, draw=none] (0,0) -- (1.5,0) arc (1:60:1.5cm) -- cycle;
    \end{scope}

    \draw (0,0) circle[radius=32mm, draw=none];

    \node[inner sep=2pt, transform shape] (AA) at (0,-3.8)  {\huge{$D^-=sl(C^-)$}};

        \foreach \i [count=\ii from 0] in {white, white, lightgray, white, lightgray, white}
    \draw (0,0) -- (\ii*60:32mm)
                    arc[start angle=\ii*60,delta angle=60]
                    -- cycle;

\end{scope}

\begin{scope}[shift={(28,0)}]
  \foreach \i [count=\ii from 0] in {white, white, white, white, white,white}
    \filldraw[fill=\i] (0,0) -- (\ii*60:32mm)
                    arc[start angle=\ii*60,delta angle=60]
                    -- cycle;

    \begin{scope}[rotate=60]
        \filldraw[fill=lightgray, draw=none] (0,0) -- (2.1,0) arc (1:60:2.1cm) -- cycle;
    \end{scope}

   \begin{scope}[rotate=240]
        \filldraw[fill=lightgray, draw=none] (0,0) -- (1.5,0) arc (1:60:1.5cm) -- cycle;
    \end{scope}

    \draw (0,0) circle[radius=32mm, draw=none];

    \node[inner sep=2pt, transform shape] (AA) at (0,-3.8)  {\huge{$SI^-(A) = A' \otimes \overline{D^-}$}};

        \foreach \i [count=\ii from 0] in {white, white, lightgray, white, lightgray, white}
    \draw (0,0) -- (\ii*60:32mm)
                    arc[start angle=\ii*60,delta angle=60]
                    -- cycle;

\end{scope}

\end{tikzpicture}
    \caption{
    Computing the pure shortening insensitive part $SI^-$ of the language of $A$.
    }
\end{subfigure}
\caption{Similarly to Fig.~\ref{fig:lang}, $\Sigma^\omega$ is represented as a circle that is partitioned into equivalence classes of word (represented as pie slices), and longer words are assumed to be further from the center of the circle. Gray areas are inside the language of the automaton, white are outside of it. In this example six equivalence classes are chosen to cover all different possible situations within an equivalence class. }
\label{fig:compute}
\end{figure}

It can be noted that it is quite easy to adapt the steps to compute less strict versions of these ``partitions'', e.g. extracting the sublanguage $SI^\pm \cup SI^-$ can be done using equations for computing $SI^-$ (see Fig.~\ref{fig:compute}) but substituting $A$ for $A'$ in all equations. This involves fewer complement operations on automata.

Unfortunately, computing the complement of an automaton $A$ is worst case exponential in the size of $A$, and the procedure described here uses several
complementations to implement set difference: $\lang(A) \setminus \lang(B) = \lang( A \otimes \bar{B})$.
At least when the automaton $A_\varphi$ comes from an LTL property, we can compute $A_{\lnot\varphi}$ instead of computing the complement $\bar{A_\varphi}$.
The other complement operations in our procedure seem unavoidable at this stage however.
Despite this high worst case complexity, the complexity of model-checking as a whole is usually
dominated by the size of the Kripke structure, which can be substantially reduced by structural reductions.
Hence a model-checking procedure using the partition and structural reductions where possible might still be competitive with the default procedure that cannot use structural reductions at all unless the property is fully stutter insensitive.

\subsection{Experiments with Partitioning}

As preliminary experimentation we computed these partitions on the LTL formulae of Section~\ref{sec:perfsLTL}~\footnote{The total number of formulae in the MCC dataset do not match those in Table~\ref{tab:formulae} since this experiment was run with a more recent version of ITS-Tools (release $2022-11$) that can parse more model/formula pairs of this benchmark ($44252$ up from $43989$ out of $45152$ model/formula pairs in the full MCC data set).}. Results are presented in Table~\ref{tab:partition}.

\begin{table}
\begin{footnotesize}
\begin{centering}
\begin{tabular}{c||c|c||c||c|c||c|c||c|c|c|c|c}
Bench. & Total & TO+MO & SI & LI & w/ $SI^\pm$ & ShI & w/ $SI^\pm$ & LS & w/ $SI^\pm$ & w/ $SI^+$ & w/ $SI^-$ & w/ $SS$ \\
\hline
Dwyer & 55 & 0 & 32 & 13 & 13 & 9 & 9 & 1 & 1&1 &1  & 0\\
Spot & 93 & 1+0 & 63 & 17 & 17 & 11 & 11 & 2 & 2 & 2 & 2 & 1 \\
RERS & 2050 & 0 & 714 & 777 & 777 & 559 & 559 & 0 &0 &0 & 0 & 0 \\
MCC & 44252 & 328+40 & 26146 & 4996 & 4994 & 6161 & 6146 & 6581 & 6527 & 6498 & 6573 & 2799 \\
\end{tabular}
\end{centering}
\end{footnotesize}
\caption{Study of partitioning of properties of the literature. The $SI$ stutter insensitive, $LI$ lengthening insensitive, $ShI$ shortening insensitive and $LS$ length sensitive columns are the same as in Table~\ref{tab:formulae}. For $LI$ and $ShI$ we also report the number of cases where the language contains a pure stuttering insensitive part $SI^\pm$ (columns $w/ SI^\pm$), and for the length sensitive properties $LS$ we also report the number of cases in which the $SI^\pm, SI^+, SI^-$ and $SS$ parts are non empty.
TO indicates a timeout of over 15 seconds, and MO indicates a memory overflow when computing the partition.}
\label{tab:partition}
\end{table}

This table shows that:
\begin{itemize}
\item An overwhelming majority (over $99.8 \%$ in all cases, and $100 \%$ of non random formulae)  of shortening or lengthening insensitive properties contain a stutter insensitive sublanguage (column $w/ SI^\pm$). This perhaps explain the high consistency of results on non decidable instances of Fig~\ref{fig:benchmark}. This observation is also true of fully length sensitive properties $LS$, with over $99.1\%$ of these properties containing a stutter insensitive sublanguage. At least for this part of the language, strategies using structural reductions and/or partial order reduction are possible.
\item While computing the partition of a language can theoretically be prohibitively expensive, in practice in most cases it is possible to do so. We fail to compute the partition for only one property of the $Spot$ benchmark, and less than one percent of the MCC full dataset. These measures use a timeout of 15 seconds, which we consider negligible in most cases in the full model-checking procedure that has timeout of the order of several minutes.
\item Most length sensitive languages (classified as ``LS'') are actually a combination of a lengthening insensitive and a shortening insensitive part, typically with some stuttering insensitive behavior added on top (in over $99 \%$ of cases), but do not actually contain an $SS$ pure fully length sensitive part. Only $42 \%$ of the length sensitive properties of the MCC actually contain a fully sensitive $SS$ part, which only represents $\approx 6 \%$ of all model/formula pairs in the MCC (that ITS-Tools could parse). These percentages fall to practically zero on non random formulae.
\end{itemize}

Overall the results in this table are very encouraging. Studying fragments of the language using the most efficient strategy for each part seems to be a promising approach, particularly for the $LS$ properties for which we would currently default to a study on the non reduced version of the model.

\subsection{Extending the language}

This section is independent of the two previous sections, and presents a new complementary approach that can help to validate cases where there are no counter-examples.

\begin{pty}
For any language $\lang$, by extension of the closure operation on automaton, let $cl(\lang)=\{r' \in \Sigma^\omega \mid \exists r \in \lang, r' \pre r \}$. 

Given two languages $\lang_{\lnot \varphi}$ and $\lang_\KS$, and a reduction function $I$ for $\lang_\KS$ such that  $\mathit{Red}_I(\lang_\KS)$ is a reduction of $\lang_\KS$,
$$
cl(\lang_{\lnot \varphi}) \cap \mathit{Red}_I(\lang_\KS) = \emptyset \implies \lang_{\lnot \varphi} \cap \lang_\KS  = \emptyset 
$$
\end{pty}

\begin{proof}
Because $cl(\lang_{\lnot \varphi}) \cap \mathit{Red}_I(\lang_\KS) = \emptyset$, we are guaranteed that, in any stutter equivalent class $\hat{r}$, all words of the reduced language
$\mathit{Red}_I(\lang_\KS)$ in this class (if any exist) are strictly longer than any word in $cl(\lang_{\lnot \varphi})$ hence they are strictly longer than any word in $\lang_{\lnot \varphi}$. Since $\mathit{Red}_I(\lang_\KS)$ only contains words that are shorter or equal to words of the original language, $\lang_\KS$ also must contain only words that are strictly longer than any word in $\lang_{\lnot \varphi}$.
\end{proof}

This observation gives us a new semi-decision procedure (only able to prove absence of counter-examples) that uses a structurally reduced model and can be applied to arbitrary properties. Of course, it is only really relevant for lengthening insensitive or length sensitive properties as stutter insensitive properties enjoy a full decision procedure and shortening insensitive properties are not modified by the $cl$ operation (by definition).
Note that for these shortening insensitive properties, we already had established in Theorem~\ref{th:short} that when the product ($\lang_{\lnot \varphi} \cap \mathit{Red}_I(\lang_\KS)$) is empty, the result is reliable.

\subsection{Revisiting the decision procedure}

A revisited model-checking approach exploiting the new features introduced in this section, and applicable to arbitrary properties could be:

\begin{enumerate}
    \item First, optimistically compute the product of the property and the reduced \KS{}.
    \item If the property language is SI, LI or ShI and the verdict is reliable, conclude.
    \item Otherwise, try to confirm the verdict obtained:
    \begin{itemize}
        \item If the product is empty, try to confirm using the product of $cl(A_{\lnot\varphi})$ and the reduced $\KS{}$. If this product is empty, conclude.
        \item If the product is not empty, try to confirm the existence of a reliable counter-example in the product of $SI^\pm \cup SI^+$ with the reduced $\KS$. If it exists conclude.  
    \end{itemize}
    \item If these tests fail, we can still try to compute the result for the $SI^\pm$ part of the language on the reduced system (since this is a full and reliable decision procedure) and proceed (if the product is empty, otherwise we have a trusted counter-example) to verify the remaining parts of the language on the original $\KS$.
\end{enumerate}

Overall for length sensitive properties, separate analysis of the parts of the language $SI^\pm, SI^+, SI^-$ on the reduced $\KS$ (using the most appropriate decision procedure) remains possible, thus extending the scope of properties where at least part of the analysis can be performed using structural reductions.

While the outlined approach has not yet been implemented, the experiments on partitioning of Table~\ref{tab:partition} indicate that there is a lot of room for these strategies to work (over $99 \%$ of properties in the benchmark contain an $SI^\pm$ part).
Moreover the high correlation ($95 \%$) between untrusted verdicts of the semi-decision procedure (see Fig.~\ref{fig:benchmark} B) and D)) and actual verdicts seems to indicate that an approach seeking to confirm untrusted verdicts could be very successful.

\section{Related Work}
\label{sec:related}

\noindent\textbf{Partial order vs structural reductions.}
Partial order reduction (POR)~\cite{porBook,Valmari90,Peled94,GodefroidW94} is a very popular approach to combat state explosion for 
\emph{stutter insensitive} formulae.
These approaches use diverse strategies (stubborn sets, ample sets, sleep sets\ldots) to consider only a subset of events at each step of the model-checking while still ensuring that at least one representative of each stutter equivalent class of runs is explored.
Because the preservation criterion is based on equivalence classes of runs, this family of approaches is limited only to the stutter insensitive fragment of LTL (see Fig.\ref{fig:lang}). 
However the structural reduction rules used in this paper are compatible and can be stacked with POR when the formula is stutter insensitive; this is the setting in which most structural reduction rules were originally defined.
Hence in the approach where we study either fully $SI$ formulas or simply the $SI^\pm$ fragment of the language of the property separately, POR can be stacked on top of our approach for studying this fragment.

\noindent\textbf{Structural reductions in the literature.}
The structural reductions rules we used in the performance evaluation are defined on Petri nets where 
the literature on the subject is rich~\cite{Berthelot85,PPP00,EHPP05,HPP06,berthomieu19,YTM20}.
However there are other formalism where similar reduction rules have been defined such as~\cite{PR08} using ``atomic'' blocks in Promela,
transaction reductions for the widely encompassing intermediate language PINS of LTSmin~\cite{Laarman18}, and even in the context of multi-threaded programs~\cite{FlanaganQ03}. All these approaches are structural or syntactic, and they are run prior to model-checking per se.

\noindent\textbf{Non structural reductions in the literature.} Other strategies have been proposed that instead of structurally reducing the system, dynamically build an abstraction of the Kripke structure where less observable stuttering occurs. These strategies build a $\KS{}$ whose language \emph{is} a reduction of the language of the original $\KS{}$ (in the sense of Def.~\ref{def:reduction}), that can then be presented to the emptiness check algorithm with the negation of the formula. They are thus also compatible with the approach proposed in this paper.
Such strategies include the Covering Step Graph (CSG) construction of~\cite{Vernadat97} where a ``step'' is performed (instead of firing a single event) that includes several independent transitions. 
The Symbolic Observation Graph of~\cite{KP08} is another example where states of the original $\KS{}$ are computed (using BDDs) and aggregated as long as the atomic proposition values do not evolve; in practice it exhibits to the emptiness check only shortest runs in each equivalence class hence it is a reduction.

\section{Conclusion}

To combat the state space explosion problem that LTL model-checking meets, structural reductions have been proposed that syntactically compact the model so that it exhibits fewer interleavings of non observable actions.  
Prior to this work, all of these approaches were limited to the stutter insensitive fragment of the logic.
We bring a semi-decision procedure that widens the applicability of these strategies to  formula that are \emph{shortening insensitive} or \emph{lengthening insensitive}. 
The experimental evidence presented shows that the fragment of the logic covered by these new categories is quite useful in practice. 
An extensive measure using the models, formulae and the two best tools of the model-checking competition 2021 shows that our strategy can improve the decision power of state of the art tools, and confirm that in the best case an exponential speedup of the decision procedure can be attained. 
We further proposed in this paper approaches that seek to confirm untrusted semi-decision verdicts using a partition of the language, and an approach based on extending the (arbitrary) property language to obtain a procedure that could reliably conclude that the product is empty while working with a reduced system.
We also identified several other strategies that are compatible with our approach since they construct a reduced language. 

Further work involves implementation and experimentation of the approaches introduced in Section~\ref{sec:extension}, as well as developing approaches based on studying the precise (untrusted) counter-example if one is produced to progressively refine the analysis in a CEGAR manner\cite{cegar00}.
\bibliographystyle{alphaurl}
\bibliography{biblio.bib}

\newcommand{\etalchar}[1]{$^{#1}$}
\begin{thebibliography}{TMRPAP24}

\bibitem[Ber85]{Berthelot85}
G{\'{e}}rard Berthelot.
\newblock Checking properties of nets using transformation.
\newblock In {\em Applications and Theory in {P}etri Nets}, volume 222 of {\em LNCS}, pages 19--40. Springer, 1985.

\bibitem[BLBDZ19]{berthomieu19}
Bernard Berthomieu, Didier Le~Botlan, and Silvano Dal~Zilio.
\newblock {Counting {P}etri net markings from reduction equations}.
\newblock {\em {International Journal on Software Tools for Technology Transfer}}, April 2019.

\bibitem[CGJ{\etalchar{+}}00]{cegar00}
Edmund~M. Clarke, Orna Grumberg, Somesh Jha, Yuan Lu, and Helmut Veith.
\newblock Counterexample-guided abstraction refinement.
\newblock In {\em {CAV}}, volume 1855 of {\em LNCS}, pages 154--169. Springer, 2000.

\bibitem[Cou99]{CouvreurFM99}
Jean{-}Michel Couvreur.
\newblock On-the-fly verification of linear temporal logic.
\newblock In {\em World Congress on Formal Methods}, volume 1708 of {\em Lecture Notes in Computer Science}, pages 253--271. Springer, 1999.

\bibitem[DAC98]{dwyer.98.fmsp}
Matthew~B. Dwyer, George~S. Avrunin, and James~C. Corbett.
\newblock Property specification patterns for finite-state verification.
\newblock In Mark Ardis, editor, {\em Proceedings of the 2nd Workshop on Formal Methods in Software Practice (FMSP'98)}, pages 7--15. ACM Press, March 1998.
\newblock \href {https://doi.org/10.1145/298595.298598} {\path{doi:10.1145/298595.298598}}.

\bibitem[DJJ{\etalchar{+}}12]{tapaal12}
Alexandre David, Lasse Jacobsen, Morten Jacobsen, Kenneth~Yrke J{\o}rgensen, Mikael~H. M{\o}ller, and Ji{\v{r}}{\'i} Srba.
\newblock {TAPAAL} 2.0: Integrated development environment for timed-arc {P}etri nets.
\newblock In Cormac Flanagan and Barbara K{\"o}nig, editors, {\em Tools and Algorithms for the Construction and Analysis of Systems}, pages 492--497, Berlin, Heidelberg, 2012. Springer Berlin Heidelberg.

\bibitem[DR18]{dureja.18.tacas}
Rohit Dureja, , and Kristin~Yvonne Rozier.
\newblock More scalable {LTL} model checking via discovering design-space dependencies ({$D^3$}).
\newblock In {\em Proceedings of the 24th International Conference on Tools and Algorithms for the Construction and Analysis of Systems (TACAS'24)}, pages 309--327, Cham, 2018. Springer International Publishing.
\newblock \href {https://doi.org/10.1007/978-3-319-89960-2_17} {\path{doi:10.1007/978-3-319-89960-2_17}}.

\bibitem[EH00]{etessami.00.concur}
Kousha Etessami and Gerard~J. Holzmann.
\newblock Optimizing {B\"u}chi automata.
\newblock In C.~Palamidessi, editor, {\em Proceedings of the 11th International Conference on Concurrency Theory (Concur'00)}, volume 1877 of {\em LNCS}, pages 153--167, Pennsylvania, USA, 2000. Springer-Verlag.

\bibitem[EHP05]{EHPP05}
Sami Evangelista, Serge Haddad, and Jean{-}Fran{\c{c}}ois Pradat{-}Peyre.
\newblock Syntactical colored {P}etri nets reductions.
\newblock In {\em {ATVA}}, volume 3707 of {\em LNCS}, pages 202--216. Springer, 2005.

\bibitem[FQ03]{FlanaganQ03}
Cormac Flanagan and Shaz Qadeer.
\newblock A type and effect system for atomicity.
\newblock In {\em {PLDI}}, pages 338--349. {ACM}, 2003.

\bibitem[GW94]{GodefroidW94}
Patrice Godefroid and Pierre Wolper.
\newblock A partial approach to model checking.
\newblock {\em Inf. Comput.}, 110(2):305--326, 1994.

\bibitem[HJM{\etalchar{+}}21]{rers21}
Falk Howar, Marc Jasper, Malte Mues, David Schmidt, and Bernhard Steffen.
\newblock The rers challenge: towards controllable and scalable benchmark synthesis.
\newblock {\em International Journal on Software Tools for Technology Transfer}, pages 1--14, 06 2021.
\newblock \href {https://doi.org/10.1007/s10009-021-00617-z} {\path{doi:10.1007/s10009-021-00617-z}}.

\bibitem[HP06]{HPP06}
Serge Haddad and Jean{-}Fran{\c{c}}ois Pradat{-}Peyre.
\newblock New efficient {P}etri nets reductions for parallel programs verification.
\newblock {\em Parallel Processing Letters}, 16(1):101--116, 2006.

\bibitem[KBG{\etalchar{+}}21]{mcc:2021}
F.~Kordon, P.~Bouvier, H.~Garavel, L.~M. Hillah, F.~Hulin-Hubard, N.~Amat., E.~Amparore, B.~Berthomieu, S.~Biswal, D.~Donatelli, F.~Galla, , S.~{Dal Zilio}, {P. G.} Jensen, C.~He, D.~{Le Botlan}, S.~Li, , J.~Srba, Y.~Thierry-Mieg, A.~Walner, and K.~Wolf.
\newblock {Complete Results for the 2021 Edition of the Model Checking Contest}.
\newblock {http://mcc.lip6.fr/2021/results.php}, June 2021.

\bibitem[KP08]{KP08}
Kais Klai and Denis Poitrenaud.
\newblock {MC-SOG:} an {LTL} model checker based on symbolic observation graphs.
\newblock In Kees~M. van Hee and R{\"{u}}diger Valk, editors, {\em Applications and Theory of {P}etri Nets, 29th International Conference, {PETRI} {NETS} 2008, Xi'an, China, June 23-27, 2008. Proceedings}, volume 5062 of {\em LNCS}, pages 288--306. Springer, 2008.

\bibitem[KPR04]{KazhamiakinPR04}
Raman Kazhamiakin, Marco Pistore, and Marco Roveri.
\newblock Formal verification of requirements using {SPIN:} {A} case study on web services.
\newblock In {\em {SEFM}}, pages 406--415. {IEEE} Computer Society, 2004.

\bibitem[Laa18]{Laarman18}
Alfons Laarman.
\newblock Stubborn transaction reduction.
\newblock In {\em {NFM}}, volume 10811 of {\em LNCS}, pages 280--298. Springer, 2018.

\bibitem[Lam90]{Lamport90theorem}
Leslie Lamport.
\newblock A theorem on atomicity in distributed algorithms.
\newblock {\em Distributed Comput.}, 4:59--68, 1990.

\bibitem[Lip75]{Lipton75}
Richard~J. Lipton.
\newblock Reduction: {A} method of proving properties of parallel programs.
\newblock {\em Commun. {ACM}}, 18(12):717--721, 1975.

\bibitem[LS89]{Lamport1989pretending}
Leslie Lamport and Fred~B. Schneider.
\newblock Pretending atomicity.
\newblock {\em SRC Research Report 44}, May 1989.
\newblock URL: \url{https://www.microsoft.com/en-us/research/publication/pretending-atomicity/}.

\bibitem[MD15]{ADL15}
Thibaud Michaud and Alexandre Duret{-}Lutz.
\newblock Practical stutter-invariance checks for {\(\omega\)}-regular languages.
\newblock In {\em {SPIN}}, volume 9232 of {\em LNCS}, pages 84--101. Springer, 2015.

\bibitem[Pel94]{Peled94}
Doron~A. Peled.
\newblock Combining partial order reductions with on-the-fly model-checking.
\newblock In {\em {CAV}}, volume 818 of {\em LNCS}, pages 377--390. Springer, 1994.

\bibitem[PP00]{PPP00}
Denis Poitrenaud and Jean{-}Fran{\c{c}}ois Pradat{-}Peyre.
\newblock Pre- and post-agglomerations for {LTL} model checking.
\newblock In {\em {ICATPN}}, volume 1825 of {\em LNCS}, pages 387--408. Springer, 2000.

\bibitem[PPH96]{porBook}
Doron~A. Peled, Vaughan~R. Pratt, and Gerard~J. Holzmann, editors.
\newblock {\em Partial Order Methods in Verification, Proceedings of a {DIMACS} Workshop, 1996}, volume~29 of {\em {DIMACS} Series in Discrete Mathematics and Theoretical Computer Science}. {DIMACS/AMS}, 1996.

\bibitem[PPR08]{PR08}
Christophe Pajault, Jean{-}Fran{\c{c}}ois Pradat{-}Peyre, and Pierre Rousseau.
\newblock Adapting {P}etri nets reductions to {P}romela specifications.
\newblock In {\em {FORTE}}, volume 5048 of {\em LNCS}, pages 84--98. Springer, 2008.

\bibitem[PPRT22]{forte22}
Emmanuel Paviot{-}Adet, Denis Poitrenaud, Etienne Renault, and Yann Thierry{-}Mieg.
\newblock {LTL} under reductions with weaker conditions than stutter invariance.
\newblock In {\em {FORTE}}, volume 13273 of {\em Lecture Notes in Computer Science}, pages 170--187. Springer, 2022.

\bibitem[SB00]{somenzi.00.cav}
Fabio Somenzi and Roderick Bloem.
\newblock Efficient {B\"u}chi automata for {LTL} formul{\ae}.
\newblock In {\em Proceedings of the 12th International Conference on Computer Aided Verification (CAV'00)}, volume 1855 of {\em LNCS}, pages 247--263, Chicago, Illinois, USA, 2000. Springer-Verlag.

\bibitem[SV12]{schewe.12.atva}
Sven Schewe and Thomas Varghese.
\newblock Tight bounds for the determinisation and complementation of generalised {B\"{u}}chi automata.
\newblock In Supratik Chakraborty and Madhavan Mukund, editors, {\em Proceedings of the 10th International Symposium on Automated Technology for Verification and Analysis (ATVA'12)}, volume 7561 of {\em Lecture Notes in Computer Science}, pages 42--56. Springer, October 2012.
\newblock \href {https://doi.org/10.1007/978-3-642-33386-6_5} {\path{doi:10.1007/978-3-642-33386-6_5}}.

\bibitem[Thi15]{ITStools}
Yann Thierry{-}Mieg.
\newblock Symbolic model-checking using {ITS}-tools.
\newblock In {\em {TACAS}}, volume 9035 of {\em LNCS}, pages 231--237. Springer, 2015.

\bibitem[Thi20]{YTM20}
Yann Thierry{-}Mieg.
\newblock Structural reductions revisited.
\newblock In {\em Petri Nets}, volume 12152 of {\em LNCS}, pages 303--323. Springer, 2020.

\bibitem[TMRPAP24]{scico24}
Yann Thierry-Mieg, Etienne Renault, Emmanuel Paviot-Adet, and Denis Poitrenaud.
\newblock A model-checker exploiting structural reductions even with stutter sensitive ltlimage 1.
\newblock {\em Science of Computer Programming}, 235, 2024.
\newblock \href {https://doi.org/10.1016/j.scico.2024.103089} {\path{doi:10.1016/j.scico.2024.103089}}.

\bibitem[Val90]{Valmari90}
Antti Valmari.
\newblock A stubborn attack on state explosion.
\newblock In {\em {CAV}}, volume 531 of {\em LNCS}, pages 156--165. Springer, 1990.

\bibitem[Var07]{Vardi07}
Moshe~Y. Vardi.
\newblock Automata-theoretic model checking revisited.
\newblock In {\em {VMCAI}}, volume 4349 of {\em LNCS}, pages 137--150. Springer, 2007.

\bibitem[VM97]{Vernadat97}
Fran{\c{c}}ois Vernadat and Fran{\c{c}}ois Michel.
\newblock Covering step graph preserving failure semantics.
\newblock In Pierre Az{\'{e}}ma and Gianfranco Balbo, editors, {\em Application and Theory of {P}etri Nets 1997, 18th International Conference, {ICATPN} '97, Toulouse, France, June 23-27, 1997, Proceedings}, volume 1248 of {\em LNCS}, pages 253--270. Springer, 1997.

\bibitem[Yan08]{yan.08.lmcs}
Qiqi Yan.
\newblock Lower bounds for complementation of omega-automata via the full automata technique.
\newblock {\em Logical Methods in Computer Science}, 4(1), March 2008.

\end{thebibliography}

\end{document}